\documentclass[a4paper, 11pt]{article}
\usepackage[usenames]{xcolor}
\usepackage{xifthen}
\usepackage{algorithmic}
\usepackage{comment}
\usepackage{tcolorbox}
\usepackage{relsize}
\usepackage[font=footnotesize]{caption}
\usepackage[normalem]{ulem}
\usepackage{bm}
\usepackage{makecell}
\usepackage{enumitem}

\def\fontsettingup{2} 

\usepackage{amsmath, amsthm, amssymb}
\usepackage[T1]{fontenc}
\ifthenelse{\fontsettingup = 1}{ \usepackage{eulerpx, eucal, tgpagella}   }{}
\ifthenelse{\fontsettingup = 2}{ \usepackage{mathpazo, eufrak, tgpagella} }{}

\definecolor{bleudefrance}{rgb}{0.19, 0.55, 0.91}
\newcommand{\prat}[1]{\textcolor{bleudefrance}{[\textbf{P:} #1]}}
\renewcommand{\prat}[1]{}

\usepackage{hyperref}
\hypersetup{
  colorlinks=true,
  linkcolor=blue,
  filecolor=magenta,
  urlcolor=cyan,
  citecolor=blue
}

\usepackage[linesnumbered,boxed,ruled,vlined]{algorithm2e}
\usepackage{tikz}
\usepackage{cleveref}
\usepackage{makecell}

\usepackage[a4paper]{geometry}
\newgeometry{
textheight=9in,
textwidth=6.5in,
top=1in,
headheight=14pt,
headsep=25pt,
footskip=30pt
}

\newtheorem{theorem}{Theorem}

\newtheorem*{claim*}{Claim}

\newtheorem{lemma}[theorem]{Lemma}
\newtheorem{proposition}[theorem]{Proposition}
\newtheorem{corollary}[theorem]{Corollary}
\theoremstyle{definition}

\newtheorem{definition}[theorem]{Definition}
\newtheorem{remark}[theorem]{Remark}
\newtheorem*{remark*}{Remark}

\ifthenelse{\fontsettingup = 1}{
  \def\*#1{\mathbf{#1}} 
  \def\+#1{\mathcal{#1}} 
  \def\-#1{\mathrm{#1}} 
  \def\^#1{\mathbb{#1}} 
  \def\!#1{\mathfrak{#1}} 
}{}

\ifthenelse{\fontsettingup = 2}{
  \def\*#1{\boldsymbol{#1}} 
  \def\+#1{\mathcal{#1}} 
  \def\-#1{\mathrm{#1}} 
  \def\^#1{\mathbb{#1}} 
  \def\!#1{\mathfrak{#1}} 
}{}

\def\oPr{\mathbf{Pr}}
\renewcommand{\Pr}[2][]{ \ifthenelse{\isempty{#1}}
  {\oPr\left[#2\right]}
  {\oPr_{#1}\left[#2\right]} } 

\def\oE{\mathbb{E}}
\newcommand{\E}[2][]{ \ifthenelse{\isempty{#1}}
  {\oE\left[#2\right]}
  {\oE_{#1}\left[#2\right]} }

\DeclareMathOperator*{\oVar}{\mathbf{Var}}
\newcommand{\Var}[2][]{ \ifthenelse{\isempty{#1}}
  {\oVar\left[#2\right]}
  {\oVar_{#1}\left[#2\right]} }

\def\oEnt{\mathbf{Ent}}
\newcommand{\Ent}[2][]{ \ifthenelse{\isempty{#1}}
  {\oEnt\left[#2\right]}
  {\oEnt_{#1}\left[#2\right]} }

\renewcommand{\epsilon}{\varepsilon}
\renewcommand{\emptyset}{\varnothing}

 \newcommand{\eps}{\varepsilon}

\newcommand{\lap}{\texttt{Lap}}

\let\epsilon=\varepsilon

\usepackage{tcolorbox}
\tcbuselibrary{skins,breakable}
\tcbset{enhanced jigsaw}

\newenvironment{tbox}{\begin{tcolorbox}[
		enlarge top by=5pt,
		enlarge bottom by=5pt,
		 breakable,
		 boxsep=0pt,
                  left=4pt,
                  right=4pt,
                  top=10pt,
                  arc=0pt,
                  boxrule=1pt,toprule=1pt,
                  colback=white
                  ]
	}
{\end{tcolorbox}}

\theoremstyle{definition}
\newtheorem{mdalg}{Algorithm}

\newcommand{\DLap}{\mathrm{DLap}}
\newcommand{\Bin}{\mathrm{Bin}}
\newcommand{\N}{\mathbb{N}}
\newcommand{\EE}{\mathbb{E}}

\newcommand{\PPr}{\mathrm{Pr}}

\newcommand{\degest}{\ensuremath{\widetilde{\deg}}\xspace}
\newcommand{\Xlow}{\ensuremath{\X^{(\ell)}}\xspace}
\newcommand{\Xhigh}{\ensuremath{\X^{\text{(h)}}}\xspace}
\newcommand{\X}{\ensuremath{\mathrm{X}}\xspace}
\newcommand{\OPT}{\ensuremath{\mathrm{OPT}}\xspace}

\newcommand{\C}{\ensuremath{\mathrm{C}}\xspace}

\newcommand{\Chigh}{\ensuremath{\C^{\text{(h)}}}\xspace}

\newcommand{\val}{\text{val}}
\newcommand{\cC}{\mathcal{C}}
\let\epsilon=\varepsilon

\newcommand{\card}[1]{\left\vert{#1}\right\vert}

\title{The Price of Privacy For Approximating Max-CSP}

\author{%
\begin{tabular}{ccc}
\begin{tabular}[t]{c}
Prathamesh Dharangutte\footnote{\href{mailto:prathamesh.d@rutgers.edu}{prathamesh.d@rutgers.edu}. Supported by NSF IIS-2229876, DMS-2220271, DMS-2311064, CCF-2208663, CCF-2118953 and CNS-2515159.}\\ Rutgers University  \\
\end{tabular} &
\begin{tabular}[t]{c}
Jingcheng Liu \footnote{
\href{mailto:liu@nju.edu.cn}{liu@nju.edu.cn}, \href{mailto:zou.zongrui@smail.nju.edu.cn}{zou.zongrui@smail.nju.edu.cn}. State Key Laboratory for Novel Software Technology, New Cornerstone Science Laboratory. Supported by the National Science Foundation of China under Grant No. 62472212. 
} \\ Nanjing University
\end{tabular} & 
\begin{tabular}[t]{c}
Pasin Manurangsi\footnote{\href{mailto:pasin@google.com}{pasin@google.com}. } \\ Google Research \\
\end{tabular} \\
\rule{0pt}{5ex}
\begin{tabular}[t]{c}
Akbar Rafiey\footnote{\href{mailto:ar9530@nyu.edu}{ar9530@nyu.edu}.} \\ New York University\\
\end{tabular}  &
\begin{tabular}[t]{c}
Phanu Vajanopath\footnote{\href{mailto:phanu.vajanopath@cs.uni.wroc.pl}{phanu.vajanopath@cs.uni.wroc.pl}.} \\ University of Wroc\l{}aw\\
\end{tabular} &
\begin{tabular}[t]{c}
Zongrui Zou\footnotemark[2] \\ Nanjing University\\
\end{tabular}
\end{tabular}
}

\begin{document}
\date{}
\pagenumbering{arabic}
\maketitle

\begin{abstract}
We study approximation algorithms for Maximum Constraint Satisfaction Problems (Max-CSPs) under differential privacy (DP) where the constraints are considered sensitive data. Information-theoretically, we aim to classify the best approximation ratios possible for a given privacy budget $\eps$. In the high-privacy regime ($\eps \ll 1$), we show that any $\eps$-DP algorithm cannot beat a random assignment by more than $O(\eps)$ in the approximation ratio. We devise a polynomial-time algorithm which matches this barrier under the assumptions that the instances are bounded-degree and triangle-free. Finally, we show that one or both of these assumptions can be removed for specific CSPs (such as Max-Cut or Max $k$-XOR) albeit at the cost of computational efficiency.




\end{abstract}

\section{Introduction}\label{sec:intro}

Maximum Constraint Satisfaction Problems (Max-CSPs) play a fundamental role in modeling and solving a wide range of complex real-world challenges, from scheduling and planning to configuration and optimization~\cite{cheng1997applying, brailsford1999constraint, salido2008introduction, ghedira2013constraint}. Given $n$ Boolean variables $x_1,x_2,\cdots, x_n \in \{-1,+1\}$, an instance of a constraint satisfaction problem is a set of $m$ constraints, where each constraint is a predicate applied to a constant-size subset of these variables. The task is to assign each variable a Boolean value that maximizes the number of constraints satisfied. We say that a randomized algorithm achieves $\rho$-approximation if it outputs an assignment whose expected number of satisfied constraints is at least $\rho$ times that of the optimum. Max-CSPs are among the most well-studied problems in the approximation-algorithms literature (e.g.~\cite{GoemansW94,KarloffZ97,Hastad01,KhotKMO07,SamorodnitskyT09,Raghavendra08,RaghavendraS09a,Chan13}).
An example of Max-CSP is the Max-$k$XOR problem where each constraint is of the form $\prod_{j=1}^k x_{i_j}=y$ for $y\in\{-1,+1\}$. A special case of Max-$k$XOR is Max-Cut: given an undirected graph, find a set of vertices $S$ that cuts as many edges as possible, where an edge is cut if exactly one of its endpoints belongs to $S$. Max-Cut can be viewed as a special case of Max-2XOR where the right-hand side is always $-1$.


Given its importance and wide-ranging applications, ensuring the privacy of sensitive data involved in Max-CSP is of particular significance, particularly in domains such as personal data processing. For example, some notable examples of CSPs are equivalent to graph problems--such as Max-Cut, and graph data including social networks can usually be very sensitive and thus designing private problems on these has already drawn growing attention in recent years. One of the main privacy notions that has been proposed and used extensively in such graph analysis is {\em differential privacy} (DP)~\cite{dwork2006calibrating}, which has been successfully applied across numerous graph problems. Key applications include generating private approximations of graph cuts~\cite{gupta2012iterative, eliavs2020differentially, DBLP:conf/nips/DalirrooyfardMN23,liu2024optimal, chandra2024differentially, fan2026private}, spectral properties~\cite{blocki2012johnson, arora2019differentially, upadhyay2021differentially}, correlation and hierarchical clustering~\cite{bun2021differentially, DBLP:conf/nips/Cohen-AddadFLMN22, DBLP:conf/icml/ImolaEMCM23, cohen2022scalable, Deng0U0Z25}, as well as in the release of various numerical graph statistics~\cite{DBLP:conf/tcc/KasiviswanathanNRS13, DBLP:conf/nips/UllmanS19, ding2021differentially, imola2022differentially, suppakitpaisarn2025counting}.

DP is a rigorous mathematical framework designed to provide privacy guarantees for personal data. Formally, for a privacy budget parameter $\eps \geq 0$, a randomized algorithm $\mathcal{A}$ is $\epsilon$-\emph{differentially private}\footnote{The algorithmic results and the weighted packing lower bound use pure DP. The fixed-predicate multiplicative lower bound, Theorem~\ref{thm:mul_hardness_csp}, is also stated for approximate DP and incurs an additive $\delta$ term.} ($\eps$-DP) if, for every possible output $o$ and any two ``neighboring'' inputs $\Phi$ and $\Phi'$, we have
\[
\Pr{\mathcal{A}(\Phi) = o} \leq e^{\epsilon} \cdot \Pr{\mathcal{A}(\Phi') = o}.
\]

The definition of ``neighboring'' datasets is critical to this framework, as it determines what information is considered sensitive and must be protected. To align with practical applications and the standard notion of \textit{edge-level} privacy in graphs~\cite{hardt2012beating, gupta2012iterative, blocki2012johnson, dwork2014algorithmic, arora2019differentially,  eliavs2020differentially, upadhyay2021differentially, liu2024optimal, aamand2025breaking, fan2026private}, we consider the constraints themselves to be the ``private information'' in this paper, while the number of variables is public. That is, two CSP instances are \emph{neighboring} if one can be obtained from the other by adding or removing an arbitrary constraint. Despite its critical importance, there is a striking lack of literature on deploying differential privacy for CSPs under such privacy notion, or even for simple special cases like Max-Cut\footnote{There is a series of work~\cite{blocki2012johnson, arora2019differentially, eliavs2020differentially, upadhyay2021differentially, liu2024optimal, liu2025almost, aamand2025breaking} that aim to generate a private synthetic graph which preserves the sizes of all cuts in the graph. Clearly, we can then compute the max-cut on such synthetic graphs by any non-private algorithm as post-processing. However, even when allowing a multiplicative error, all these methods fail to provide an additive error bound on all-cut approximation better than $\tilde{O}(n^{5/4}/\epsilon)$, which is worse than simply running the Exponential mechanism.}. 

Nevertheless, a trivial approximation algorithm for any CSP is to output a uniformly random assignment. We use $\mu$ to denote the fraction of the constraints satisfied by such a random assignment in expectation. For example, in the Max-$k$XOR problem (including Max-Cut), a random assignment satisfies an expected fraction of $\mu = \frac{1}{2}$ of the constraints. This $\mu$-approximation algorithm is already \emph{perfect} in privacy (i.e., satisfies DP with privacy budget $\epsilon = 0$) as it does not require any information about the constraints. Setting computational complexity aside, the other extreme occurs when $\epsilon$ tends to infinity, allowing us to achieve arbitrarily close to $1$-approximation. A naturally interesting but previously unexplored problem is to classify the best approximation ratio for CSP in terms of the privacy budget $\epsilon$ beyond these two extremes of no privacy ($\epsilon \to \infty$) and perfect privacy ($\epsilon = 0$). This is the main question our work aims to address:

\begin{center}
\begin{quote}
    \emph{What is the best approximation ratio for Max-CSPs achievable by $\eps$-DP algorithms?}
\end{quote}
\end{center}

Before we continue, we stress that the above question is interesting even \emph{without} any computational complexity consideration. Indeed, our impossibility results (i.e. upper bounds on the approximation ratios) will be all information theoretic (via privacy requirement). Nevertheless, as highlighted below, some of our algorithms actually run in polynomial time. 

\subsection{Our Contributions}

We make some initial progresses towards answering the aforementioned questions. Specifically, we identify two regimes of $\eps$ which results in different approximation ratios:
\begin{itemize}
\item \textbf{High-Privacy Regime} ($\eps \ll 1$): In this regime, we show that the approximation ratio can be at most $\mu + O(\eps)$. In other words, the ``advantage'' we can get over the random assignment is at most $O(\eps)$. We devise a polynomial-time algorithm matching this bound for bounded-degree and triangle-free instances. Finally, we show that these assumptions can sometimes be removed for Max-Cut and Max-$k$XOR, although the algorithms either become inefficient or achieve worse advantage over the random assignment. Our results for this regime are summarized in \Cref{tab:summary}.
\item \textbf{Low-Privacy Regime} ($\eps \gg 1$): We show here that we can nearly achieve the optimum, with a ``deficit'' of $\Theta(1/\eps)$. Furthermore, the best approximation ratio is no better than $1 - \Theta(1/e^\eps)$. 
\end{itemize}
We next explain these results in more detail, starting with the simplest setting: low-privacy.

\subsubsection{Low-Privacy Regime}

On the algorithmic (lower bound on the approximation ratio) side, perhaps the most direct approach is to use the well known \emph{Exponential Mechanism (EM)}~\cite{mcsherry2007mechanism} that samples an assignment of $n$ Boolean variables according to a specific distribution decided by the constraints. Standard analysis yields an \emph{additive} approximation error of $O(n/\epsilon)$ compared to the optimal number of satisfiable constraints. This does not yet translate to any approximation ratio, i.e. when the number of constraints is $o(n)$. Nevertheless, we observe that, if there is any isolated variable (with no constraints), then the resulting distribution from EM (on the remaining variables) remains the same. As such, we may assume w.l.o.g. that no variable is isolated, implying that there are $\Omega(n)$ constraints. This, together with the aforementioned additive error bound, gives the following:
%
%
\begin{corollary}[Section~\ref{sec:preliminary_dp}]
\label{cor:large_eps}
For any $\epsilon > 0$, there is an $\epsilon$-differential privacy, $(1-O(1/\epsilon))$-approximation algorithm for Max-CSPs.
\end{corollary}

Complementing this, we also present the optimal ``deficit'' we can hope for:

\begin{theorem}[Informal version of \Cref{thm:mul_hardness_csp}, Section~\ref{sec:mul_hardness}]
For any $\eps > 0$ and $0\leq \delta < 1$,  there is no $(\epsilon,\delta)$-DP algorithm that achieves $(1 + \delta -\Omega(1/e^\epsilon))$-approximation for Max-CSPs.
\end{theorem}

\subsubsection{High-Privacy Regime: General Max-CSPs}

We now move on to the more interesting yet more challenging small-$\epsilon$ regime. Recall again that we let $\mu$ to be the (expected) fraction of constraints satisfied by a random assignment, which is perfectly private. In this regime, we show that one can only get an ``advantage'' of at most $O(\eps)$ over the random assignment: 

\begin{theorem}
[Informal version of Theorem~\ref{thm:mul_hardness_csp}
, Section~\ref{sec:mul_hardness}]\label{thm:intro_upper_bound} For any $\eps > 0$ and $0\leq \delta < 1$,  there is no $(\epsilon,\delta)$-DP algorithm that achieves $(\mu + \delta+O(\epsilon))$-approximation for Max-CSPs.
\end{theorem}

Unfortunately, we do not obtain a general algorithm that achieves a matching $(\mu + \Omega(\eps))$ for all Max-CSP instances. Nevertheless, we obtain such bounds on special instances. To do this, we turn to the approximation algorithm literature, where there is a wealth of techniques to beat random assignments~\cite{shearer1992note, haastad2002advantage, khot2007linear, BarakMORRSTVWW15}. Specifically, Barak et al.~\cite{BarakMORRSTVWW15} give a polynomial-time (non-private) $\left(\mu + \Omega\left(\frac{1}{\sqrt{D}}\right)\right)$-approximation algorithm for $D$-bounded-degree\footnote{An instance of CSP is $D$-bounded-degree if each variable appears in at most $D$ constraints.} and triangle-free\footnote{An instance of CSP is \textit{triangle-free} if no two distinct constraint scopes intersect in more than one variable, and no three distinct constraints have their three pairwise scope intersections witnessed by three distinct variables (\Cref{def:triangle-free}).} instances. We give a DP version of this algorithm, as stated below, and thereby match the bound in \Cref{thm:intro_upper_bound} for such instances. Throughout this paper, all our proposed mechanisms ensure privacy for any (possibly empty) CSP instance. The structural constraints, such as triangle freeness or bounded degree, are required only for the utility analysis.

\begin{theorem}
    [Restatement of Lemma~\ref{lem:x_j-priv-guarantee} and Lemma~\ref{lem:bounded-deg-tri-free-guarantee}, Section~\ref{sec:triangle-free-bounded-degree-csp}]\label{thm:intro_csp_bounded} 
    For any $0 < \epsilon\leq 1$ and $D \in \mathbb{N}_+$, there is a \textbf{polynomial-time} $\eps$-DP algorithm that achieves a $\left(\mu + \Omega\left(\frac{\epsilon}{\sqrt{D}}\right)\right)$-approximation for any triangle-free, $D$-bounded-degree Max-CSP instance where every constraint has arity at least 2.
\end{theorem}

It remains an intriguing open question as to whether one can remove the triangle-freeness and bounded-degree assumptions and achieve $\left(\mu + \Omega(\eps)\right)$-approximation on all instances. Note, however, that such an algorithm \emph{must be inefficient} since Max-CSPs for many predicates are known to be NP-hard to approximate to within $(\mu + o(1))$ factor (aka \emph{approximation resistant}), see, e.g.~\cite{Chan13}.

\subsubsection{High-Privacy Regime: Max-Cut and Max-$k$XOR}

While we cannot resolve the aforementioned questions, we will show that the triangle-freeness and/or bounded-degree assumptions can be removed for Max-Cut and Max-$k$XOR, although at the cost of the efficiency and sometimes the advantage as well. In particular, all algorithms stated below run in exponential time. Most of the results below are obtained via a unified framework that reduces general instances to bounded-degree ones using private partitioning and then combining the bounded-degree algorithm above with a careful application of the Exponential Mechanism.

For Max-$k$XOR when $k$ is an \emph{odd} integer, we manage to remove \emph{both} the triangle-freeness and the bounded-degree assumptions. The latter is through our aforementioned framework, while the former is through generalizing the argument from \cite{BarakMORRSTVWW15} in the non-private case. Thus, we can achieve the tight advantage for \emph{all} instances of Max-$k$XOR when $k$ is odd:

\begin{theorem}
[Informal version of Corollary~\ref{cor:max-kxor-oddk-main}]\label{thm:intro-kxor-oddk}
Fix an odd $k\geq3$. There are constants $B_k, c_k, \eps_0(k) > 0$ such that, for every $0 < \eps \leq \eps_0(k)$, every Max-$k$XOR instance with with pairwise distinct scopes and $m \geq B_k/\eps^2$ constraints admits an $\eps$-DP algorithm that achieves a $(1/2 + c_k\eps)$-approximation.
\end{theorem}

For every fixed $k\geq2$, without relying on oddness, we can still remove the degree bound at the cost of requiring triangle-freeness. The resulting advantage is $\Omega_k(\eps^3)$.

\begin{theorem}
[Informal version of Theorem~\ref{thm:triangle-free}, Section~\ref{sec:triangle-free-max-kxor}]\label{thm:intro-kxor-triangle-free} 
For every fixed $k\geq2$, there are constants $\eps_0(k), c_k > 0$ such that, for $0 < \eps \leq \eps_0(k)$, an $\eps$-DP algorithm achieves a $(1/2 + c_k\eps^3)$-approximation for triangle-free Max-$k$XOR.
\end{theorem}

For Max-Cut (a special case of Max-2XOR), a more specialized private version of Shearer's algorithm~\cite{shearer1992note} improves this to the optimal-order advantage on triangle-free graphs.

\begin{theorem}
[Informal version of~\Cref{thm:privacy_private_maxcut,thm:utility_maxcut_unbounded}, Section~\ref{sec:maxcut-unbounded}]\label{thm:intro_maxcut_unbounded} There is a constant $\eps_0>0$ such that, for every $0<\eps\leq\eps_0$, an $\eps$-DP algorithm achieves a $(1/2+\Omega(\eps))$-approximation for Max-Cut on every triangle-free graph.
\end{theorem}

More interestingly, we can remove the triangle-freeness assumption for Max-Cut while incurring only a very minor loss on the approximation ratio, which establishes an nearly tight utility bound. We achieve this by randomly subsampling the graph's edges with a probability proportional to $\epsilon$ to extract an approximate matching, which bypasses the approximation bottlenecks caused by triangles. Although the matching only covers a small fraction of the graph, it is sufficient to secure the desired approximation: non-matching edges can easily contribute the uniform baseline of $1/2$, leaving the matching to provide the essential $\Omega(\epsilon)$ global advantage (see also a detailed explanation in Section~\ref{sec:tech_overview_removing-triangle-freeness}).

\begin{theorem}
[Informal version of~\Cref{thm:maxcut_without_trianglefreeness}, Section~\ref{sec:max-cut_without_trianglefreeness}]\label{thm:intro_maxcut_trianglefreeness} For all sufficiently small $\eps>0$, every graph with $m=\Omega(\log(1/\eps)/\eps^2)$ edges admits an $\eps$-DP algorithm with approximation ratio $1/2+\Omega(\eps/\log(1/\eps))$.
\end{theorem}

\begin{table}[t]
\centering
\begin{tabular}{|c|c|c|c|c|c|}
\hline
Problem & Upper Bound & Lower Bound & Triangle-Free & Bounded-Degree \\
\hline
Max-CSPs & $\mu +{O}\left( \eps\right)$ & $\mu { + }\Omega\left(\frac{\epsilon}{\sqrt{D}} \right)$ & Required & Required \\
\hline
Max-$k$XOR (odd $k$) & $\frac{1}{2} +{O}\left( \eps\right)$ & $\frac{1}2{ + }\Omega\left(\eps \right)$ & Not Required & Not Required \\
\hline
Max-$k$XOR ($k\in \mathbb{N}_+$) & $\frac{1}{2} +{O}\left( \eps\right)$ & $\frac{1}2{ + }\Omega\left(\eps^3 \right)$ & Required & Not Required\\
\hline
Max-Cut  & $\frac{1}{2} +{O}\left( \eps\right)$  & $\frac{1}{2} +  {\Omega}\left({\eps}\right)$ & Required & Not Required\\
\hline
Max-Cut  & $\frac{1}{2} +{O}\left( \eps\right)$  & $\frac{1}{2} +  {\Omega}\left({\frac{\eps}{\ln(1/\eps)}}\right)$ & Not Required & Not Required\\
\hline
\end{tabular}
\caption{Summary of results for the high-privacy regime. The last two columns show the assumptions required in our algorithms (i.e. lower bounds).} 
\label{tab:summary}
\end{table}

\section{Technical Overview}
\label{sec:overview}

Here, we introduce the technical ingredients that underpin our results on private CSP approximation with multiplicative error, and outline the key ideas behind our privacy and utility analysis. To improve readability, we first illustrate our approach in the graph setting, where the goal is to approximate the Max-Cut with an approximation ratio better than $1/2$. This special case already captures the core of our technical ideas. We then discuss how these ideas extend to general CSPs.

We present the technical elements by first considering the simplest yet most restrictive setting, assuming bounded degree and triangle-freeness. We then describe how both assumptions can be relaxed and outline the mechanism for general graphs, which achieves nearly optimal utility.

\subsection{Private Approximation on Bounded Degree Triangle-Free Graph}
As a warm up, we begin by introducing Shearer's algorithm for Max-Cut~\cite{shearer1992note} and discussing how to privatize it via a simple method. This approach gives an $\epsilon$-DP algorithm that on any triangle-free graph with bounded degree $D$, it outputs a cut with an expected size of at least $\left(\frac{1}{2} + \Omega\left(\frac{\epsilon}{\sqrt{D}}\right)\right) \text{OPT}_G$ for $\epsilon = O(1)$. 

The algorithm for generating a large bipartite subgraph of any input graph $G$ in~\cite{shearer1992note} proceeds as follows: first, the vertices are partitioned into two sets $A$ and $B$ uniformly at random. A vertex is classified as ``good'' if more than half of its neighbors lie in the opposite set; if exactly half of its neighbors lie in the opposite set, it is labeled good with probability $\frac{1}{2}$; otherwise, it is labeled ``bad''. The good vertices remain in their initial sets, while the bad vertices are randomly reassigned to either $A$ or $B$. The resulting bipartite subgraph $X$ consists of all edges between $A$ and $B$ after this reassignment. When the input graph $G$ is triangle-free with bounded maximum degree $D$, it can be shown that the expected number of edges in $X$ is $(1/2 + \Omega(1/\sqrt{D}) \text{OPT}_G$ where $\text{OPT}_G$ is the size of the maximum cut in $G$.

Fix any $0\leq \epsilon\leq O(1)$. In the context of edge-level privacy,  a pair of graphs are neighboring if they differ by at most one edge. Then, the only step in which the randomized algorithm needs private information is to decide whether a vertex is ``good'', which depends on the topology of the graph. Therefore, one can verify that simply using the Randomized Response mechanism on the outcome of whether a vertex is good or bad (i.e., flip the outcome with probability $\frac{1}{e^\eps + 1}$) would yield a $(1/2 + \Omega(\epsilon/\sqrt{D}))$-approximation with privacy.

This can also be extended to the general CSP setting. To see this, we recall the algorithm for triangle-free CSPs in \cite{BarakMORRSTVWW15}.  They start by first randomly partitioning the variables into two sets: \(F\) (Fixed) and \(G\) (Greedy). The variables in \(F\) are assigned uniformly at random. Then, for each variable \(x_j\) in \(G\), it considers \emph{active constraints} (those with exactly one variable in \(G\)) involving \(x_j\). For each such constraint, it computes the derivative \(Q_\ell\) (with respect to \(x_j\)) evaluated on the already fixed variables. The algorithm then sets \(x_j = \text{sgn}(\sum Q_\ell - \theta_j)\), where \(\theta_j\) is a median of the random variable \(\sum Q_\ell\), ensuring that \(x_j\) remains unbiased (uniformly \(\pm 1\)). This greedy assignment is designed to gain a significant advantage on active constraints while preserving randomness on inactive constraints due to the triangle-free property. The random partition ensures that, in expectation, the algorithm beats a random assignment by \(\Omega(1/\sqrt{D})\). A similar application of the Randomized Response mechanism by flipping the sign of each $x_j$ with some probability depending on $\epsilon$ yields a private solution while preserving the advantage.

\paragraph{Getting  better approximation from a more sophisticated analysis.} Recall that in Shearer's algorithm~\cite{shearer1992note}, each vertex determines whether it is ``good'' based on whether the majority of its neighbors share the same assignment. The assignment of each vertex is chosen uniformly at random. Now, consider adding or removing an edge incident to a vertex \( u \) with degree \( d(u) \). If the number of neighbors of \( u \) with the same assignment exceeds those with a different assignment by at least $2$ (or vise versa), then this edge modification does not affect whether \( u \) is classified as good or bad. This observation allows us to construct a natural coupling between the assignments of the same vertex in two neighboring graphs. With probability at least \( 1 - O\left( 1/{\sqrt{d(u)}} \right) \), the status of \( u \) as ``good'' remains consistent across both graphs, implying no privacy loss. Leveraging this fact, we are able to establish an approximation ratio of $1/2 + \Theta(\frac{1}{\sqrt{D  + 1/\epsilon^2}})$.


\subsection{Removing the Dependency on Degrees}\label{sec:remove_bounded_degree}

Now, suppose that we have a black-box $\epsilon$-DP oracle that, for any triangle-free graph $G$ of bounded degree $D$, returns a cut achieving an approximation ratio of $1/2 + \Theta\left(\frac{1}{\sqrt{D  + 1/\epsilon^2}}\right)$ for any $\epsilon \geq 0$. 
Unfortunately, this oracle alone is insufficient, as in the worst case, the guarantee can be as bad as $1/2 + \Theta\left(\frac{1}{\sqrt{n}}\right)$ for $D = \Theta(n)$.

To eliminate the dependence on $D$ or $n$ and obtain an approximation guarantee that depends \emph{only} on $\epsilon$, 
%
we consider the following partitioning strategy. First, set a threshold $d$ (to be specified later).
Then, let $V_1$ be the approximate set of all high-degree vertices with degree $> d$ found via some private thresholding algorithm, and $V_2$ be the remaining vertices. Let $G^{(h)}$ be the subgraph on $V_1$ containing all edges with both endpoints in $V_1$, and $G^{(\ell)}$ be the subgraph on $V_1 \cup V_2$ containing the remaining edges (i.e. those with at least one endpoint in $V_2$). 

We wish to use the private oracle mentioned above on the graph\footnote{Note that vertices in $G^{(\ell)}$ may have high degree. However, a more careful analysis of Shearer's algorithm~\cite{shearer1992note} shows that it suffices to obtain the desired $1/2 + O(1/\sqrt{D})$ approximation ratio as long as \emph{one endpoint} of every edge has bounded degree. This weaker condition is satisfied by $G^{(\ell)}$. In Section~\ref{sec:csp}, we show that this property generalizes naturally to CSPs.} $G^{(\ell)}$, while handling $G^{(h)}$ with a separate algorithm, e.g. the Exponential Mechanism. However, a challenge arises: $G^{(h)}$ and $G^{(\ell)}$ share the vertices in $V_1$. 
Consequently, the labels assigned to shared vertices by the two independent private mechanisms may be inconsistent. A dedicated merging procedure is necessary to resolve these conflicts and combine the two partial assignments. In particular, we consider the following merging trick. With probability $1/2$, we proceed as follows:
\begin{enumerate}
    \item Use the exponential mechanism to find an assignment for the vertices in $V_1$ that (approximately) maximizes the size of the cut in the subgraph $G^{(h)}$. The scoring function for this mechanism is simply the size of the cut induced by the assignment.
    \item For each vertex in $V_2$, assign its label uniformly and independently at random (without any greedy selection).
\end{enumerate}
With the remaining probability $1/2$, we simply run the private oracle for bounded-degree graphs on the full graph $G$. To see how this fair-coin merging balances the utility losses, we set the degree threshold to $d = \Theta(1/\epsilon^2)$. 

Now we suppose $\text{OPT}_{G^{(h)}} = \Omega(m)$. For the high-degree subgraph $G^{(h)}$, every vertex has a degree of at least $d$, limiting its size to at most $O(m/d) = O(m\epsilon^2)$ vertices. Running the Exponential Mechanism on this restricted space incurs an additive privacy loss of only $O(m\epsilon)$. In the dense case where $\text{OPT}_{G^{(h)}} = \Omega(m)$, this loss can be bounded by $O(\epsilon)\text{OPT}_{G^{(h)}}$. Therefore, the expected cut from this branch is at least $(1 - O(\epsilon))\text{OPT}_{G^{(h)}} + m^{(\ell)}/2$, where $m^{(\ell)}$ is the number of edges outside $G^{(h)}$. On the other hand, the private oracle run on the full graph guarantees an $\Omega(1/\sqrt{d})$ advantage on the low-degree edges $G^{(\ell)}$ while performing no worse than random overall. This gives an expected cut of at least $\frac{1}{2}\text{OPT}_G + \Omega(1/\sqrt{d})\text{OPT}_{G^{(\ell)}}$. By substituting $1/\sqrt{d} = \Theta(\epsilon)$ and taking the expectation over the fair coin toss, the expected global approximation ratio is lower-bounded by:
\[
\frac{1}{2}\left((1 - O(\epsilon))\text{OPT}_{G^{(h)}} + \frac{m^{(\ell)}}{2}\right) + \frac{1}{2}\left(\frac{1}{2}\text{OPT}_G + \Omega(\epsilon)\text{OPT}_{G^{(\ell)}}\right).
\]
Since $\text{OPT}_G \le \text{OPT}_{G^{(h)}} + m^{(\ell)}$, the term $\frac{1}{4}\text{OPT}_{G^{(h)}} + \frac{m^{(\ell)}}{4}$ from the first branch covers $\frac{1}{4}\text{OPT}_G$, abd the remaining $\frac{1}{4}\text{OPT}_{G^{(h)}}$ dominates the $O(\epsilon)\text{OPT}_{G^{(h)}}$ loss for small $\epsilon$. Summing everything up, the $\Omega(\epsilon)$ advantage from the low-degree subgraph and the near-optimal contribution from the high-degree subgraph combine together, yielding a global $(1/2 + \Omega(\epsilon))\text{OPT}_G$ approximation.

Alternatively, suppose $\text{OPT}_{G^{(h)}} = o(m)$. This implies that $\Omega(m)$ edges in the graph must have at least one low-degree endpoint. Consequently, the private oracle (the one having guarantee on bounded-degree graphs we developed earlier) run on the full graph directly yields an $\Omega(m/\sqrt{d}) = \Omega(m\epsilon)$ global advantage over a random assignment. Since the exponential mechanism on the high-degree subgraph performs no worse than random, averaging the two outcomes over the fair coin toss similarly establishes the same global $(1/2 + \Omega(\epsilon))\text{OPT}_G$ approximation in this case as well.

\subsection{Removing the Assumption on Triangle-Freeness}\label{sec:tech_overview_removing-triangle-freeness}
Our main technical highlight is the removal of the triangle-freeness assumption on the input instances, an assumption that is typically essential for Shearer-type greedy algorithms. To achieve this, we adopt a different partitioning strategy and design a more sophisticated procedure for handling the ``bounded-degree'' subgraph. For simplicity\footnote{In the actual proof we instead set the threshold $d = O(\ln(1/\epsilon)/\epsilon)$ to ensure the accuracy of private thresholding. Therefore, the final guarantee has an extra $O(\ln(1/\epsilon))$ factor, as stated in Theorem~\ref{thm:intro_maxcut_trianglefreeness}.}, consider using the threshold $d = O(1/\epsilon)$ to separate high-degree vertices (denoted $V_1$) from low-degree vertices (denoted $V_2$) in the graph. This leads to three cases:

\begin{enumerate}
    \item \textbf{High crossing edges:} If there are already many edges (e.g., at least $0.6m$) between $V_1$ and $V_2$, then after applying private thresholding, simply outputting $(V_1, V_2)$ as the cut yields a good approximation of Max-Cut.

    \item \textbf{Dense subgraph within $V_1$:} Otherwise, one of the two induced subgraphs by $V_1$ or $V_2$ must contain a lot of edges, i.e., $E(V_1) = \Omega(m)$ or $E(V_2) = \Omega(m)$, where $E(V_i)$ denotes the set of edges with both endpoints in $V_i$. We first assume $E(V_1) = \Omega(m)$, and let $\text{OPT}_{V_1}$ denote the size of the maximum cut in the subgraph induced by $V_1$. Since the threshold is $d = O(1/\epsilon)$, then if we have a sufficiently accurate private thresholding, the average degree in $V_1$ is at least $\Omega(d)$. This implies
    \[
        \text{OPT}_{V_1} \ge E(V_1)/2 = \Omega(|V_1| d) = \Omega(|V_1|/\epsilon).
    \]
    Applying the Exponential Mechanism to $V_1$ ensures that the expected additive error relative to $\text{OPT}_{V_1}$ is at most $O(|V_1|/\epsilon)$, which is within a constant factor of $\text{OPT}_{V_1}$. Therefore, by carefully choosing constants, we can obtain a cut on $V_1$ of expected size at least $(1/2 + c)\text{OPT}_{V_1}$. Extending this cut by letting the vertices in $V_2$ choose their side randomly provides sufficient advantage over a random assignment because
    \[
        c\,\text{OPT}_{V_1} = \Omega(m) = \Omega(\text{OPT}_G),
    \]
    where $\text{OPT}_G$ denotes the size of the maximum cut in the input graph.
    \item \textbf{Sparse subgraph within $V_2$:} Now, assume that $E(V_2) = \Omega(m)$, and let $\text{OPT}_{V_2}$ denote the maximum cut in the induced subgraph on $V_2$. Again, with sufficiently accurate private thresholding, this subgraph has a bounded degree at most $O(d) = O(1/\epsilon)$. Directly applying the Exponential Mechanism for Max-Cut on $V_2$ offers no guarantee on the approximation ratio, as $\text{OPT}_{V_2}$ could be much smaller than $|V_2|/\epsilon$. However, interestingly, if we independently subsample each edge with probability $p = O(\epsilon)$, we obtain two key observations:

\begin{itemize}
    \item The surviving edges from the subsampling would {almost} form a \textit{matching} of size $\Omega(m\epsilon)$, since $E(V_2) = \Omega(m)$ and each vertex has degree at most $d = O(1/\epsilon) = O(1/p)$\footnote{If the surviving graph is not a matching, we can randomly delete an $O(1)$-fraction of edges without affecting privacy amplification.}.
    \item Running the Exponential Mechanism for Max-Cut on this matching with privacy budget $\epsilon' = \epsilon / p = O(1)$ ensures overall $\epsilon$-DP (for $0\leq \epsilon \leq 1$), by the principle of privacy amplification via sub-sampling~\cite{balle2018privacy}.
\end{itemize}

Since the matching contains at most $\Theta(m\epsilon)$ non-isolated vertices, the expected additive error of the Exponential Mechanism is at most $O(m\epsilon / \epsilon') = O(m\epsilon)$. Note that the size of the maximum cut in the matching equals the number of edges (i.e., $\Omega(m\epsilon)$), so this error represents at most a constant-fraction loss. Thus, by carefully calibrating the constants, we could ensure that  we already obtain an ``advantage'' of $\Omega(m\epsilon)$ over half for edges in the matching. 

While for the edges in $E(V_2)$ that do not survive from the sub-sampling, suppose the resulting graph forms a matching, each such edge must connect two disconnected components. Consequently, the placement of its endpoints, as determined by the Exponential Mechanism, is entirely \textit{independent}, which implies that the contribution of these edges to the cut is at least as large as what would be obtained by a uniformly random assignment.
 Finally, assigning vertices in $V \setminus V_2$ randomly yields a cut of expected size
\[
\underbrace{m/2}_{\substack {\text{each edge is contributing}\\ \text{no worse than random}}} + \underbrace{\Omega(\epsilon m)}_{\substack{\text{the ``advantage'' from}\\ \text{ matching edges}}} \ge (1/2 + \Omega(\epsilon))\,\text{OPT}_G.
\]
\end{enumerate}

This concludes our high-level overview of our algorithms and analyses.

\section{Preliminaries}

\subsection{Constraint Satisfaction Problem}\label{sec:preliminary_csp}

An instance $\Phi$ of Max-CSP consists of $n$ Boolean variables
$x_1,\ldots,x_n\in\{\pm1\}$ and a finite multiset of $m$ constraints.
Constraint $\ell$ is a pair $(P_\ell,S_\ell)$, where
$S_\ell=(i_1,\ldots,i_{r_\ell})$ is an ordered tuple of distinct indices,
$1\le r_\ell\le k$, and
$P_\ell:\{\pm1\}^{r_\ell}\to\{0,1\}$.  It is satisfied by $x$ when
$P_\ell(x_{i_1},\ldots,x_{i_{r_\ell}})=1$.  Throughout, we assume that $k$ is a
constant, and without the loss of generality we use the standard irredundant
convention: every predicate is nonconstant and depends on every coordinate
listed in its scope.  Constant constraints can be tracked separately, and
variables on which a predicate does not depend are removed from its scope.

For Max-$k$XOR, every scope has size $k$ and constraint $\ell$ has a sign
$b_\ell\in\{\pm1\}$; its predicate is
\[
P_\ell(x_{S_\ell})
  =\frac{1+b_\ell\prod_{i\in S_\ell}x_i}{2}.
\]
The degree $\deg_\Phi(x_i)$ is the number of constraints whose scope contains
$i$.  We omit the subscript when the instance is clear.  A maximum-degree
bound $D$ is imposed only when explicitly stated.

One particular class of Max-CSP instances of interest in this paper is that of triangle-free instances. The formal definition is as follows. 

\begin{definition}
\label{def:triangle-free}
    A Max-CSP instance $\Phi$ is \emph{triangle-free} if the following two
    conditions hold.  First, the scopes of any two distinct constraints
    intersect in at most one variable.  Second, there do not exist three distinct constraints
    $a,b,c$ and three distinct variables $u,v,w$ such that
    $u\in S_a\cap S_b$, $v\in S_b\cap S_c$, and
    $w\in S_c\cap S_a$.
\end{definition}

\paragraph{Fourier representation.}

We write the Fourier expansion of a Boolean function $f: \{\pm1\}^n \rightarrow \mathbb{R}$ as
\[
f(x) = \sum_{S \subseteq [n]} \widehat{f}(S) \chi_S(x) \qquad \text{where} \; \chi_S(x):=\prod_{i\in S}x_i.
\]

We define its influence as $\text{Inf}_i[f]=\sum_{S\ni i}\widehat f(S)^2 = \E[x]{(\partial_i f)(x)^2}$. Here $\partial_i f$ is the derivative of $f$ with respect to the $i$th coordinate. This can be defined by $(\partial_i f)(x')=\tfrac{f(x',+1)-f(x',-1)}{2}$, where $(x',b)$ denotes $(x_1,\dots,x_{i-1},b,x_{i+1},\dots,x_n)$. 
For a given instance $\Phi$ and assignment $x \in \{\pm1\}^n$, the value of assignment $x$, denoted by $\text{val}_{\Phi}(x)$, is the number of satisfied constraints, i.e. $\sum_{\ell} P_\ell (x_{S_\ell})$. This is a multilinear polynomial of degree at most $k$. For an instance $\Phi$, let $\mu=m^{-1}\sum_{\ell=1}^m\E[z\sim\mathrm{Unif}(\{\pm1\}^{r_\ell})]{P_\ell(z)}$ be the expected fraction of constraints satisfied by a uniform assignment.  We define the advantage of an assignment $x$ over a uniform assignment by
\begin{equation}
\label{eq:associated-polynomial}
\mathfrak{P}(x) = \frac{1}{m} \sum_{\ell = 1}^{m} \bar{P}_\ell(x_{S_\ell})
\end{equation}
where $\bar{P}_\ell=P_\ell-\E[z]{P_\ell(z)}$. Thus,
$\text{val}_{\Phi}(x)=\left(\mu+\mathfrak{P}(x)\right)m$.

\subsection{Differential Privacy}\label{sec:preliminary_dp}

Differential privacy is a widely accepted notion of privacy.  Let $\mathcal{U}$
be a record universe and let $\mathcal{U}^*$ denote the collection of finite
multisets of records.  In our applications a record is one constraint (or one
edge), while the variable/vertex universe is public.

\begin{definition}
    [Neighboring dataset]
    Datasets $d,d'\in\mathcal{U}^*$ are neighboring, denoted $d\sim d'$, if
    one is obtained from the other by adding or removing one record. Thus,
    unless explicitly stated otherwise, we use unbounded (add/remove)
    adjacency. Section~\ref{sec:weighted-l1-lb} separately identifies a
    weighted $\ell_1$ model for one weighted lower bound.
\end{definition}

\begin{definition}[$(\eps, \delta)$-differential privacy] \label{def:dp}
    Let $\mathcal{M}:\mathcal{U}^*\rightarrow\mathcal{O}$ be a randomized
    algorithm. For $\eps>0$ and $\delta\in[0,1)$, $\mathcal{M}$ is
    $(\eps,\delta)$-differentially private if, for every measurable
    $S\subseteq\mathcal{O}$ and all neighboring $d,d'\in\mathcal{U}^*$,
    \[
    \Pr{\mathcal{M}(d) \in S} \leq e^\eps \cdot \Pr{\mathcal{M}(d') \in S} + \delta
    \]
\end{definition}

Below we mention some properties of differentially private algorithms.

\begin{lemma}[Post processing \cite{dwork2006calibrating}]
\label{lem:postprocess}
    Let $\mathcal{M}:\mathcal{U}^*\rightarrow \mathcal{O}$ be a
    $(\eps,\delta)$-differentially private algorithm. If
    $f:\mathcal{O}\rightarrow\mathcal{O}'$ is any data-independent randomized
    mapping, then $f\circ\mathcal{M}$ is also $(\epsilon,\delta)$-differentially
    private.
\end{lemma}

\begin{lemma}[Basic composition \cite{dwork2006calibrating}]\label{lem:composition}
     Let $\mathcal{M}_1,\ldots,\mathcal{M}_k$ be (possibly adaptive)
     mechanisms whose conditional guarantees are $(\eps_i,\delta_i)$-DP.
     Their joint release is
     $(\sum_{i=1}^k\eps_i,\sum_{i=1}^k\delta_i)$-DP.
\end{lemma}

We now describe mechanisms that are commonly used in differential privacy. We begin by defining the sensitivity of a function.

\begin{definition}
    [$\ell_p$-sensitivity] Let $f:\mathcal{U}^*\rightarrow \mathbb{R}^k$ be a query function on datasets. Its sensitivity is
    \[\mathsf{sens}_p (f) = \max_{D\sim D'} \|f(D) - f(D')\|_p.\]
\end{definition}

We now describe the Laplace mechanism which adds noise from the Laplace distribution as per the $\ell_1$ sensitivity of the function.

\begin{definition}
    [Laplace distribution] Given parameter $b$, Laplace distribution (with scale $b$) is the distribution with probability density function 
    \[\text{Lap}(x|b) = \frac{1}{2b} \exp\left(-\frac{|x|}{b}\right).\]
\end{definition}
We use $\text{Lap}(b)$ to denote the Laplace distribution with scale $b$.

\begin{lemma}[Laplace mechanism]
\label{lem:laplace-priv}
     Suppose $f:\mathcal{U}^*\rightarrow \mathbb{R}^k$ is a query function with $\ell_1$ sensitivity $\mathsf{sens}$. Then the mechanism
    $$\mathcal{M}(D) = f(D) + (Z_1,\cdots,Z_k)^\top$$
    is $(\epsilon,0)$-differentially private, where $Z_1,\cdots, Z_k$ are i.i.d. random variables drawn from $\text{Lap}(\mathsf{sens}/\epsilon)$.
\end{lemma}

The following is a standard tail bound for Laplace distribution:
\begin{lemma}
    \label{lem:laplace-tail}
    Let $X\sim\text{Lap}(b)$. Then, for every $t\geq0$,
    $\Pr{|X|\geq tb}=\exp(-t)$.  Moreover,
    $\Pr{X\geq tb}=\Pr{X\leq-tb}=\frac12\exp(-t)$.
\end{lemma}

We also introduce the privacy amplification result by Balle et al.~\cite{balle2018privacy}.

\begin{lemma}[Privacy Amplification by Poisson Subsampling \cite{balle2018privacy}]
\label{lem:privacy_amplification}
    Let $\mathcal{M}: \mathcal{U}^* \to \mathcal{O}$ be an $(\eps_0, 0)$-differentially private mechanism under add/remove adjacency. Let $\mathcal{M}'$ be the randomized mechanism that, given a dataset $D \in \mathcal{U}^*$, first constructs a subsampled dataset $D'$ by independently including each record of $D$ with probability $p \in [0, 1]$, and then outputs $\mathcal{M}(D')$. Then, under add/remove adjacency, $\mathcal{M}'$ is $(\eps, 0)$-differentially private, where
    \[
        \eps = \ln\left(1 + p(e^{\eps_0} - 1)\right).
    \]
\end{lemma}

We next state the randomized response mechanism which can be used to privately output a binary bit of information.

\begin{lemma}
    [Randomized response]
\label{lem:random-response}
    Let $x \in \{0,1\}$ be the bit we wish to output privately. Then, the randomized response mechanism that outputs $x$ with probability $\frac{e^\eps}{1+e^\eps}$ and $1-x$ with probability $\frac{1}{1+e^\eps}$ is $(\eps,0)$-DP.
\end{lemma}

Next, we describe the Exponential mechanism, which is useful when we want to choose a candidate solution from a set of solutions. Let $\mathcal{C}$ be the set of all possible candidates, and let $s:\mathcal{U}^*\times \mathcal{C}\rightarrow \mathbb{R}_{\geq 0}$ be a scoring function. We define the sensitivity of $s$ to be
\[\mathsf{sens}_s = \max_{c\in\mathcal{C}} \max_{D\sim D'} |s(D,c) - s(D',c)|.
\]

Now, to output a candidate that maximizes the scoring function, the exponential mechanism $\mathcal{E}$ for dataset $D\in \mathcal{U}^*$ is defined by the distribution

$\mathcal{E}(D) := \text{Choose a candidate } c\in \mathcal{C}$ { with probability proportional to } $\exp\left( \frac{{\epsilon \cdot s(D,c)}}{2\mathsf{sens}_s}\right)$.

We have the following guarantees for $\mathcal{E}(\cdot)$.

\begin{lemma}\cite{mcsherry2007mechanism}
\label{prop:exponential_mechanism}
    Let $\mathcal{E}(\cdot)$ be an exponential mechanism sampling in some discrete space $\mathcal{C}$. $\mathcal{E}(\cdot)$ is $(\epsilon,0)$-differentially private. Further, for any input dataset $D$ and $t>0$, let $\OPT = \max_{c\in \mathcal{C}}s(D,c)$, then we have that
    \[\Pr{s(D,\mathcal{E}(D)) \leq \OPT - \frac{2\mathsf{sens}_s}{\epsilon} (\log (|\mathcal{C}|) + t)} \leq \exp(-t).
    \]
    In particular,
    \[
    \E{s(D,\mathcal{E}(D))} \geq \OPT - \frac{2\mathsf{sens}_s}{\epsilon} (\log (|\mathcal{C}|) + 1)
    \]
\end{lemma}
\noindent Here, we present the proof of Corollary~\ref{cor:large_eps} in Section~\ref{sec:intro}, which is a direct corollary of Lemma~\ref{prop:exponential_mechanism}:
\begin{proof}[Proof of Corollary~\ref{cor:large_eps}.]
Let $\Phi$ be the input CSP instance $\Phi$ with $n$ variables, $m$ constraints, and each constraint in $\Phi$ contains at most a constant number of variables. Also let $\OPT$ be the optimal value of Max-CSP on $\Phi$. Suppose $m = \Omega(n)$, then we simply let the discrete space $\mathcal{C}$ to be $\mathcal{C} = \{-1,+1\}^n$, namely each candidate is an assignment $\mathbf{x}$ of each variable, and define the scoring function $s(\Phi, \mathbf{x})$ as the number of constraints satisfied. Clearly the sensitivity of $s(\Phi, \mathbf{x})$ is exactly one. Further, note that a uniform assignment
satisfies each constraint with probability at least $2^{-k}$, Hence
$\OPT\geq 2^{-k}m = \Omega(m)$. Thus, by Lemma~\ref{prop:exponential_mechanism}, we have 
\begin{equation}\label{eq:large_eps_1}
    \mathbb{E}[s(\Phi, \mathbf{x})] \geq \OPT - O\left(\frac{n}{\epsilon}\right) \geq \OPT - O\left(\frac{m}{\epsilon}\right) \geq (1-O(1/\epsilon))\OPT
\end{equation}
as $m = O(\OPT)$.
Otherwise, suppose $m = o(n)$. Note that $m$ constraints cover at most $O(m)$ variables and the isolated variables do not affect the satisfaction of any constraints. In particular, let $$N_1 := \{v\in [n]:v \text{ appers in at least one constraint}\},$$ and denote by $n_1 = |N_1|$. We just let the discrete space $\mathcal{C}'$ to be $\mathcal{C}' = \{-1,+1\}^{N_1}$, the possible assignments of variables covered by constraints. 
Again by Lemma~\ref{prop:exponential_mechanism}, we have
\begin{equation}\label{eq:large_eps_2}
    \mathbb{E}[s(\Phi, \mathbf{x})] \geq \OPT - O\left(\frac{n_1}{\epsilon}\right) \geq \OPT - O\left(\frac{m}{\epsilon}\right) \geq (1-O(1/\epsilon))\OPT.
\end{equation}
Combining both eq.(\ref{eq:large_eps_1}) and eq.(\ref{eq:large_eps_2}) completes the proof.
\end{proof}

\section{Private Max-CSP on Triangle-Free Instances}\label{sec:csp}

In this section, we describe private algorithms that obtain an advantage over
uniform assignments for Max-CSPs.  We begin with a polynomial-time rounding algorithm for
triangle-free Max-CSP (\Cref{alg:low-degree-triangle-free}); its most convenient
corollary is for bounded maximum degree.  We then combine that rounding with
the exponential mechanism to handle triangle-free Max-$k$XOR instances with
no degree bound (\Cref{alg:main}).
\subsection{Privacy for Triangle-Free Bounded Degree Arbitrary CSPs}\label{sec:triangle-free-bounded-degree-csp}

Fix any constant $k\in \mathbb{N}_+$ and $k\geq 2$. For privacy, we define the mechanism below on every input (possibly empty) Max-CSP instance whose constraint arities are between $2$ and $k$ (as in~\cite{BarakMORRSTVWW15}), while triangle-freeness is used \emph{only} for the analysis of utility guarantees.

We privatize the rounding algorithm in Section~5 of
\cite{BarakMORRSTVWW15}.  The irredundancy convention from the preliminaries (Section~\ref{sec:preliminary_csp}) ensures that every scoped variable has nonzero influence.  Without it, a dummy
variable could have identically zero derivative and the degree-based lower
bound below would be false.  For a given nonempty triangle-free CSP instance $\Phi$
with $m$ constraints, recall the associated polynomial
(\Cref{eq:associated-polynomial}) that represents the advantage of an
assignment $x$ over a uniform assignment,
\[
\mathfrak{P}(x) = \frac{1}{m} \sum_{\ell = 1}^{m} \bar{P}_\ell(x_{S_\ell}).
\]
The aim is to find an assignment such that $\mathfrak{P}(x)$ is large enough.

The algorithm independently places every variable into one of two sets $F$
and $G$, each with probability $1/2$.  Variables in $F$ receive independent
uniform signs, so there is no privacy leakage. Assignments to variables in $G$ depend on the constraints and
are hence need to be privatized.

Call a constraint $(P_\ell, S_\ell)$ active if exactly one variable $x_i$ with $i\in S_\ell$ is in $G$, and inactive otherwise.
For $x_j \in G$, let $N_j$ be the set of indices of active constraints whose scopes contain $j$. Let $A_j \subset F$ be the set of variables that are part of active constraints with $x_j$. Formally, $\forall x_j \in G$,
\[
N_j = \{\ell: (P_\ell, S_\ell) \; \text{is active and } j \in S_\ell\} \qquad A_j = \bigcup_{\ell \in N_j} \{x_i \in F : i \in S_\ell\} 
\]

For $\ell \in N_j$, let $\bar{P_\ell}(x_{S_\ell}) = x_j Q_\ell + R_\ell$,
with $Q_\ell = \partial_j\bar{P}_\ell$.
Since $R_\ell$ is the sum of monomials involving only the variables from $F$, its value depends solely on the uniform assignments in $F$. Consequently, its expected contribution to $\E{\mathfrak{P}(x)}$ is $0$. As a result, to bound the advantage an assignment to $x_j$ has for the active constraints , i.e. $\E{\sum_{\ell \in N_j} \bar{P_\ell}(x_{S_\ell})}$, we only need to bound 
$\E{x_j \sum_{\ell \in N_j} Q_\ell}$. Write
\[
T_j:=\sum_{\ell\in N_j}Q_\ell,
\]
where the probability below is over the uniform assignment to $A_j$. Let $\theta_j$ be a median belonging to the support of $T_j$, and define
\[
    b_j:=\Pr{T_j>\theta_j},\qquad p_j:=\Pr{T_j=\theta_j},
    \qquad \rho_j:=\frac{1/2-b_j}{p_j}.
\]
The median property gives $b_j\le 1/2\le b_j+p_j$, so $\rho_j\in[0,1]$. Define the randomized sign $Z_j$ by
\begin{equation}\label{eq:randomized-median-sign}
Z_j=
\begin{cases}
+1,&T_j>\theta_j,\\
-1,&T_j<\theta_j,\\
+1\text{ with probability }\rho_j\text{ and }-1\text{ otherwise},&T_j=\theta_j.
\end{cases}
\end{equation}
Thus $\Pr{Z_j=+1}=b_j+p_j\rho_j=1/2$; in particular, it is the randomized tie-breaking rule, rather than the median property alone, that makes $Z_j$ uniform. Moreover,
\[
\E{Z_jT_j}=\E{Z_j(T_j-\theta_j)}+\theta_j\E{Z_j}
=\E{|T_j-\theta_j|}.
\]
Because the instance is triangle-free, the random variables $Q_\ell$ for
$\ell\in N_j$ depend on disjoint sets of uniform variables.  Following the same arguments as in
\cite{BarakMORRSTVWW15} we obtain
\begin{equation}
\label{eq:exp-value-adv}
\E{Z_jT_j} \geq \exp\left(-O(k)\right) \cdot \sqrt{|N_j|}.
\end{equation}

Assignment of $x_j \in G$ depends on the constraints through the quantity
$Z_j$. We privatize this bit using randomized response.
We now describe the private algorithm.

\begin{algorithm}[H]
\caption{$(\eps,0)$-DP algorithm with a triangle-free Max-CSP utility guarantee.}
\label{alg:low-degree-triangle-free}

\begin{algorithmic}[1]

\REQUIRE A possibly empty Max-CSP instance $\Phi$ with $n$ variables and $m$
constraints of arities between $2$ and $k$, and privacy parameter
$\epsilon>0$.\\

\ENSURE Assignment $x \in \{\pm1\}^n$. 

\STATE \textbf{(Random partition).}
Independently place each variable into $F$ or $G$, each with probability $1/2$.

\STATE \textbf{(Assignment to $F$).}
Assign variables in $F$ uniformly at random.

\STATE \textbf{(Assignment to $G$).}
For each $x_i \in G$, generate $Z_i$ using the randomized median rule in \eqref{eq:randomized-median-sign}, with independent tie-breaking coins across variables. Independently let $Y_i$ be $1$ with probability $\frac{e^\eps}{1+e^\eps}$ and $-1$ otherwise. Set $x_i=Y_iZ_i$.

\STATE \textbf{(Output).}
Return the combined assignment $x\in\{\pm1\}^n$.
\end{algorithmic}
\end{algorithm}

Before proving the final guarantee we prove some properties of the private assignment $x_j \in G$ by \Cref{alg:low-degree-triangle-free} which will be crucial in our analysis. The privacy proof applies to all neighboring instances in the full domain specified above; triangle-freeness is used only for the independence and utility properties.

\begin{lemma}
\label{lem:x_j-priv-guarantee}
    For every possibly empty Max-CSP instance $\Phi$ whose constraint arities
    are between $2$ and $k$, and every privacy parameter $\eps'>0$,
    \Cref{alg:low-degree-triangle-free} is $(\eps',0)$-differentially private
    on the full add/remove input domain. If, moreover, $\Phi$ is triangle-free
    and $0<\eps'<1$, the assignments to all $x_j\in G$ satisfy the following
    guarantees:
    \begin{itemize}
        \item Each $x_j\in G$ is uniform, and its data-dependent part uses
        only the variables in $A_j$ and constraints in $N_j$.
        \item On the scope of every inactive constraint, the output signs are
        mutually independent and uniform.
        \item For each assignment $x_j \in G$ we have,
        \[
        \E{x_j \cdot \sum_{\ell\in N_j}Q_\ell} \geq \Omega(\eps')\cdot \exp\left(-O(k)\right) \cdot \sqrt{|N_j|}.
        \]
    \end{itemize}
\end{lemma}

\begin{proof}
    We first prove that $x_j \in G$ is uniformly random. By \eqref{eq:randomized-median-sign}, including its randomized rule on the event $T_j=\theta_j$, we have $\Pr{Z_j=+1}=\Pr{Z_j=-1}=1/2$. Now, for $x_j=Y_jZ_j$,
    \[
    \Pr{x_j = +1} = \frac{e^{\eps'}}{1+e^{\eps'}}\cdot\Pr{Z_j = +1} + \left(1 - \frac{e^{\eps'}}{1+e^{\eps'}}\right)\cdot\Pr{Z_j = -1} = \frac{1}{2},
    \]
    which proves the first part.  For the second part, fix an inactive
    constraint $C$.  If its scope lies in $F$, the claim is immediate.
    Otherwise, consider two distinct variables $x_j,x_{j'}\in S_C\cap G$.
    The no-overlap condition implies that neither $A_j$ nor $A_{j'}$
    intersects $S_C$.  The no-hyper-triangle condition implies that
    $A_j\cap A_{j'}=\varnothing$: otherwise $C$, an active constraint through
    $j$, and an active constraint through $j'$ would form a hyper-triangle.
    Thus the latent signs for variables in $S_C\cap G$ use disjoint uniform
    randomness and are independent of the signs in $S_C\cap F$.  The
    independent randomized-response coins preserve this independence.

    We next prove the full-domain privacy statement; this part does not use
    triangle-freeness. Condition on the data-independent partition and on all
    internal uniform signs assigned to $F$.  If the added or removed
    constraint is inactive, it is absent from every $N_j$ and does not change
    the output distribution.  If it is active, it belongs to exactly one
    $N_j$ and can change only the latent sign $Z_j$.  The released sign
    $x_j=Y_jZ_j$ is randomized response applied to that latent sign.  This
    remains $\eps'$-DP even when $Z_j$ is randomized and its tie probability
    is data-dependent: for either released sign and every distribution of
    $Z_j$, its probability lies in
    $[1/(1+e^{\eps'}),e^{\eps'}/(1+e^{\eps'})]$.  All other released
    coordinates have the same conditional distribution.  Averaging over the
    conditioned randomness proves $(\eps',0)$-DP for every neighboring pair,
    whether or not either instance is triangle-free.

    Finally, to bound the expected value, notice that $Y_j$ is independent of $Z_j$ and,
    \begin{align*}
        \E{x_j \cdot \sum_{\ell\in N_j}Q_\ell} &= \E{Y_j Z_j T_j} \\
        &\geq \frac{e^{\eps'} - 1}{e^{\eps'} + 1} \cdot \exp\left(-O(k)\right) \cdot \sqrt{|N_j|}   \tag{From \Cref{eq:exp-value-adv}}\\ 
        &\geq \Omega(\eps')\cdot \exp\left(-O(k)\right) \cdot \sqrt{|N_j|}                    \tag{For $\eps' < 1$}
    \end{align*}

\end{proof}

\paragraph{Efficient implementation of the tie-breaking rule.}
The exact median rule is well defined on every bounded-arity input: on an
arbitrary instance, its law can be computed exactly by enumerating all
assignments to $A_j$. On triangle-free instances, the rule has the following
polynomial-time implementation for Boolean predicates of constant arity; no
data-dependent sampling accuracy is needed. For every
$\ell\in N_j$,
\[
 Q_\ell=\partial_jP_\ell\in\{-1/2,0,1/2\}.
\]
Its three probabilities under a uniform assignment to
$S_\ell\setminus\{j\}$ are obtained by enumerating at most $2^{k-1}$ inputs.
Triangle-freeness makes these variables disjoint across $\ell\in N_j$, so the
law of $2T_j=\sum_{\ell\in N_j}2Q_\ell$ is the convolution of $|N_j|$
distributions supported on $\{-1,0,1\}$. A dynamic program maintains its
$O(|N_j|)$ possible integer values and their exact rational probabilities.
The probability denominators and numerators have $O(k|N_j|)$ bits, so the
convolution has polynomial bit complexity. Scanning the resulting cumulative
distribution gives exact values of $\theta_j,b_j,p_j$, and $\rho_j$.

Consequently, the implemented latent sign is exactly balanced, not merely
approximately balanced. Its computation uses only the local constraint set
$N_j$; adding or removing an unrelated constraint changes neither its law nor
its running parameters. Thus the locality and privacy proof above applies
without an additional approximation argument.

With the guarantees of \Cref{lem:x_j-priv-guarantee}, the rest of the arguments are similar to \cite{BarakMORRSTVWW15}. We include them here for completion.

\begin{lemma}
\label{lem:bounded-deg-tri-free-guarantee}
    Given a triangle-free instance $\Phi$ of Max-CSP whose constraints have
    arities between $2$ and $k$, with $m\geq1$ constraints and privacy parameter
    $0<\eps<1$, let $d_i=\deg_\Phi(x_i)$.
    \Cref{alg:low-degree-triangle-free} is $(\eps,0)$-differentially private
    and outputs an assignment $x\in\{\pm1\}^n$ such that
    \[
    \E{\mathfrak{P}(x)}
       \geq \frac{c_k\eps}{m}\sum_{i:d_i>0}\sqrt{d_i},
    \]
    for a constant $c_k>0$ depending only on $k$.  Consequently, if
    $\max_i d_i\leq D$, then
    $\E{\mathfrak P(x)}\geq c_k\eps/\sqrt D$.
    The exact convolution implementation described above runs in time
    polynomial in $n$ and $m$ for fixed $k$.
\end{lemma}

\begin{proof}
    Recall that variables are first randomly partitioned into sets $F$ and $G$. Variables in $F$ are assigned independent uniform signs in $\{\pm1\}$. A constraint $(P_\ell, S_\ell)$ is active if its scope $S_\ell$ contains exactly one variable from $G$.
    
    First, notice that each inactive constraint $(P_i, S_i)$ contributes exactly $0$ to $\E{\mathfrak{P}(x)}$. This is because the assignments to variables in $F$ are uniformly random and mutually independent (as established in \Cref{lem:x_j-priv-guarantee}). As a result, letting $\C_{\text{acv}}$ denote the set of active constraints for a specific random partition $(F, G)$, we can compute the expected advantage by summing only over active constraints and regrouping by variables in $G$:
    \begin{align*}
    \E{\mathfrak{P}(x)} &= \frac{1}{m}\cdot \E{ \sum_{(P_i, S_i) \in \C_{\text{acv}}} \bar{P_i}(x_{S_i}) } \\
    &= \frac{1}{m} \cdot\E{ \sum_{x_j\in G}\sum_{i \in N_j} x_j Q_i} \\
    &\geq \frac{1}{m} \cdot \Omega(\eps) \cdot \exp(-O(k)) \cdot \sum_{i=1}^n \Pr{x_i \in G} \cdot \E{\sqrt{|N_i|} \; \Big| \; x_i \in G}. \tag{From \Cref{lem:x_j-priv-guarantee}}
    \end{align*}

    Next, we bound the term $\E{\sqrt{|N_i|} \mid x_i \in G}$. Conditioning on $x_i \in G$, a constraint $\ell$ containing $x_i$ becomes active if and only if all its other $r_\ell - 1$ variables fall into $F$. Since the variables are partitioned independently and uniformly, this occurs with probability $2^{-(r_\ell-1)} \geq 2^{-(k-1)}$. 
    
    Crucially, because the instance is triangle-free, the scopes of any two constraints sharing $x_i$ do not intersect anywhere else. Therefore, the events of these constraints becoming active are mutually independent. This implies that, conditional on $x_i \in G$, the number of active constraints $|N_i|$ stochastically dominates a binomial random variable $\text{Bin}(d_i, p)$ with $p = 2^{-(k-1)}$. Using the standard property of binomial distributions that $\E{\sqrt{\text{Bin}(n,p)}} = \Omega(\sqrt{np})$ for any constant $p$ and $n \ge 1$, we obtain:
    \[
    \E{\sqrt{|N_i|} \mid x_i \in G} \geq c'_k \sqrt{d_i},
    \]
    where $c'_k > 0$ is a constant depending only on $k$. Since the partition is uniformly random, $\Pr{x_i \in G} = 1/2$. Substituting this back into our main inequality yields the degree-sum lower bound:
    \begin{align*}
    \E{\mathfrak{P}(x)} &\geq \frac{1}{m} \cdot \frac{\Omega(\eps) \exp(-O(k))}{2} \sum_{i:d_i>0} c'_k \sqrt{d_i} \\
    &= \frac{c_k\eps}{m} \sum_{i:d_i>0} \sqrt{d_i},
    \end{align*}
    where all constants depending on $k$ are absorbed into $c_k$. 

    Finally, we derive the corollary for bounded-degree instances. If $d_i \leq D$ for all $i$, then $\sqrt{d_i} = d_i / \sqrt{d_i} \geq d_i / \sqrt{D}$. Substituting this into the degree-sum bound gives:
    \begin{align*}
    \E{\mathfrak{P}(x)} &\geq \frac{c_k\eps}{m} \sum_{i=1}^n \frac{d_i}{\sqrt{D}} \geq \frac{c_k\eps}{m} \cdot \frac{m}{\sqrt{D}} = \frac{c_k\eps}{\sqrt{D}}.
    \end{align*}
    This completes the proof.
\end{proof}

In particular, on a maximum-degree-$D$ instance,
\Cref{alg:low-degree-triangle-free} satisfies
$(\mu+\Omega_k(\eps/\sqrt D))m$ constraints in expectation.  We next use
the stronger degree-sum form of the lemma to remove the degree bound for
triangle-free Max-$k$XOR.

\subsection{Privacy for Triangle-Free Max-$k$XOR with No Degree Bound}
\label{sec:triangle-free-max-kxor}
In this section, we prove Theorem~\ref{thm:intro-kxor-triangle-free} for triangle-free max-$k$XOR instances. In particular, we will combine the degree-sum guarantee of
\Cref{lem:bounded-deg-tri-free-guarantee} with an exponential mechanism on
privately identified high-degree variables. In Algorithm~\ref{alg:main}, we set
\[
    \tau=\frac{A_k}{\eps^4},
\]
where $A_k$ is a sufficiently large constant depending only on $k\in \mathbb{N}_+$.

\begin{algorithm}[H]
\caption{$(\eps,0)$-DP algorithm for Max-$k$XOR with no degree bound.}
\label{alg:main}

\begin{algorithmic}[1]

\REQUIRE An instance $\Phi$ of Max-$k$XOR with $n$ variables and $m$ constraints of arity at most $k$, and a privacy parameter $\epsilon \in (0, 1]$.\\

\ENSURE An assignment $x \in \{\pm 1\}^n$.

\STATE \textbf{(Private degree estimation).}
\label{line:deg-partition} 
For each variable $x_i$, obtain a private estimate of its degree: $\degest(x_i) = \deg_\Phi(x_i) + \lap\left(\frac{3k}{\eps}\right)$.

\STATE \textbf{(Instance partitioning).}
Set a threshold $\tau = A_k/\eps^4$ (where $A_k$ is a sufficiently large constant depending on $k$). Let $\Xhigh = \{x_i \mid \degest(x_i) > \tau\}$ and $\Xlow = \{x_i \mid \degest(x_i) \leq \tau\}$. Let $\Chigh$ be the set of constraints whose entire scope lies within $\Xhigh$.

\STATE \textbf{(High-degree candidate).}
\label{line:exp-mech-random} 
Run the Exponential Mechanism with privacy parameter $\eps/3$ on the variables in $\Xhigh$ to output an assignment that maximizes the number of satisfied constraints in $\Chigh$. For variables in $\Xlow$, assign them uniformly at random in $\{\pm 1\}$. Let $x^{(1)} \in \{\pm 1\}^n$ be this combined global assignment.

\STATE \textbf{(Degree-sum candidate).}
\label{line:bounded-deg} 
Run \Cref{alg:low-degree-triangle-free} with privacy parameter $\eps/3$ on the \textbf{full} instance $\Phi$. Let $x^{(2)} \in \{\pm 1\}^n$ be the resulting assignment.

\STATE \textbf{(Output).}
Output $x^{(1)}$ with probability $1/2$; otherwise, output $x^{(2)}$.

\end{algorithmic}
\end{algorithm}

\begin{theorem}
\label{thm:triangle-free}
For every fixed $k\geq2$ and every $0<\eps\leq1$,
\Cref{alg:main} is $(\eps,0)$-differentially private on
(possibly empty) Max-$k$XOR instances. Moreover,
there are constants $\eps_0,c_k>0$ such that, for every nonempty triangle-free
Max-$k$XOR instance $\Phi$ and every $0<\eps\leq\eps_0$, the algorithm returns
an assignment $x$ satisfying
\[
    \E{\val_\Phi(x)}\geq
       \left(\frac12+c_k\eps^3\right)\OPT_\Phi .
\]
The algorithm is exponential-time because of the exponential mechanism in
the high-degree branch.
\end{theorem}

\begin{proof}[Proof of \Cref{thm:triangle-free}]
We first prove privacy. The vector of degrees has $\ell_1$-sensitivity at
most $k$ under constraint add/remove adjacency, so its noisy release is
$(\eps/3,0)$-DP. Condition on any fixed output of this first mechanism. The set
$\Xhigh$ is then fixed, and adding or removing one constraint changes the score
of each assignment in the fixed universe $\{\pm1\}^n$ by at most one.
Consequently, the second step is $(\eps/3,0)$-DP. The degree-sum candidate is
also $(\eps/3,0)$-DP by
\Cref{lem:x_j-priv-guarantee}. Adaptive composition (\Cref{lem:composition}) and
postprocessing (\Cref{lem:postprocess}) prove $(\eps,0)$-DP.

We record a baseline that holds for every realized private partition. Let $V_h$ be the set of variables that appear in at least one constraint of $\Chigh$. On
$\Chigh$, the Exponential Mechanism re-weights the uniform distribution to favor assignments with higher scores. Specifically, its expected score on $\Chigh$ is lower-bounded by that of a uniform random assignment:
\[
    \E{\val_{\Chigh}(x^{(1)})\mid\Xhigh} = \sum_{y\in \{\pm 1\}^{|V_h|}} \val_{\Chigh}(y) \frac{\exp\left(\frac{\eps}{3}\val_{\Chigh}(y)\right)}{\sum_{z} \exp\left(\frac{\eps}{3}\val_{\Chigh}(z)\right)} \geq \frac{|\Chigh|}{2},
\]
because the weights $\exp(\eps \cdot \val_{\Chigh}(\cdot) / 3)$ are monotonically increasing with respect to the score. Every constraint outside
$\Chigh$ contains a uniformly assigned variable in $\Xlow$ and hence is
satisfied with probability exactly $1/2$, even after conditioning on all other
variables. Thus, letting $m_h = |\Chigh|$ and $m_\ell = m - m_h$, we have
\begin{equation}\label{eq:x1-half-revised}
    \E{\val_\Phi(x^{(1)})\mid\Xhigh} = \E{\val_{\Chigh}(x^{(1)})\mid\Xhigh} + \frac{m_\ell}{2} \geq \frac{m_h}{2} + \frac{m_\ell}{2} = \frac{m}{2}.
\end{equation}
Likewise, because every $k$-XOR predicate has uniform baseline $1/2$,
\Cref{lem:bounded-deg-tri-free-guarantee} gives
\begin{equation}\label{eq:x2-degree-sum-revised}
    \E{\val_\Phi(x^{(2)})}
       \geq \frac m2+c'_k\eps\sum_i\sqrt{d_i},
       \qquad d_i:=\deg_\Phi(x_i),
\end{equation}
where the factor $1/3$ in the privacy budget is absorbed into $c'_k$.
This run uses fresh randomness, so the same expectation holds after
conditioning on the independently generated private degree partition.

Let $H$ be the set of constraints all of whose variables have true degree at
least $2\tau$. We distinguish two cases.

\smallskip
\noindent\textbf{Case I: $|H|<m/2$.}
Let $J=\{i:d_i<2\tau\}$. Every constraint outside $H$ contains a variable in
$J$, and hence $\sum_{i\in J}d_i\geq m/2$. Since
$\sqrt{d_i}\geq d_i/\sqrt{2\tau}$ for $i\in J$,
\[
 \sum_i\sqrt{d_i}\geq\frac{m}{2\sqrt{2\tau}}.
\]
Combining \eqref{eq:x1-half-revised} and \eqref{eq:x2-degree-sum-revised}, the expected value over the final fair coin is
\begin{align*}
 \E{\val_\Phi(x)} &= \frac{1}{2}\E{\val_\Phi(x^{(1)})} + \frac{1}{2}\E{\val_\Phi(x^{(2)})} \\
 &\geq \frac{m}{4} + \frac{1}{2}\left(\frac m2+\Omega_k(\eps^3 m)\right) \\
 &= \frac m2+\Omega_k(\eps^3 m) \geq \left(\frac12+\Omega_k(\eps^3)\right)\OPT_\Phi.
\end{align*}

\smallskip
\noindent\textbf{Case II: $|H|\geq m/2$.}
Let $L$ be the set of constraints containing a variable of true degree at
most $\tau/2$. The Laplace tail (Lemma~\ref{lem:laplace-tail}) gives
\begin{align*}
 \oPr\{x_i\notin\Xhigh\mid d_i\geq2\tau\}
   &\leq \tfrac12\exp(-\tau\eps/(3k)),\\
 \oPr\{x_i\in\Xhigh\mid d_i\leq\tau/2\}
   &\leq \tfrac12\exp(-\tau\eps/(6k)).
\end{align*}
Choose $A_k$ sufficiently large. Uniformly for $0<\eps\leq1$, the first
quantity is at most $10^{-3}/k$ and the second is at most
$10^{-4}\eps^4$. Let
\[
 M_H=|H\setminus\Chigh|
 \quad\text{and}\quad
 F_L=|L\cap\Chigh|.
\]
A union bound within each constraint and linearity of expectation give
$\E{M_H}\leq10^{-3}m$ and
$\E{F_L}\leq10^{-4}\eps^4m$. Applying Markov's inequality to these
nonnegative loss counts shows that, with
probability at least $0.98$,
\begin{equation}\label{eq:good-private-partition}
 M_H\leq m/8
 \qquad\text{and}\qquad
 F_L\leq\eps^4m.
\end{equation}

Condition on this good event. Variables in $V_h$ of true degree greater
than $\tau/2$ number at most $2km/\tau$. Every remaining variable in $V_h$
occurs in a constraint counted by $F_L$, so there are at most $kF_L$ such
variables. Therefore
\begin{equation}\label{eq:high-variable-count-revised}
 r_h := |V_h| \leq\frac{2km}{\tau}+kF_L=O_k(m\eps^4).
\end{equation}
Let $\OPT_h$ be the optimum of the instance induced by $\Chigh$. A uniform
assignment satisfies $m_h/2$ of these constraints in expectation, and hence
$\OPT_h\geq m_h/2\geq3m/16$ (since $m_h \geq |H|-M_H \geq 3m/8$). The exponential mechanism guarantee and
\eqref{eq:high-variable-count-revised} imply
\begin{align*}
 \E{\val_{\Chigh}(x^{(1)})\mid\Xhigh}
 &\geq \OPT_h-\frac{6(r_h\log2+1)}{\eps} \\
 &\geq \OPT_h-O_k(m\eps^3) \tag{since $m\geq2\tau \implies \frac{1}{\eps} = O(\eps^3 m)$} \\
 &\geq \OPT_h-O_k(\eps^3\OPT_h) \\
 &\geq \left(1-O_k(\eps^3)\right)\OPT_h.
\end{align*}
By decreasing the constant $\eps_0=\eps_0(k)$ if necessary, the last display is at least $(1/2+\Omega_k(\eps^3))\OPT_h$.

Conditional on the good partition, the high-degree candidate, the degree-sum candidate, and the final fair coin therefore give
\begin{align*}
 \E{\val_\Phi(x)\mid\Xhigh}
 &= \frac12\E{\val_\Phi(x^{(1)})\mid\Xhigh} + \frac12\E{\val_\Phi(x^{(2)})\mid\Xhigh} \\
 &= \frac12 \left( \E{\val_{\Chigh}(x^{(1)})\mid\Xhigh} + \frac{m_\ell}{2} \right) + \frac12\E{\val_\Phi(x^{(2)})\mid\Xhigh} \\
 &\geq \frac12\left[\left(\frac12+\Omega_k(\eps^3)\right)\OPT_h
               +\frac{m_\ell}{2}\right]
 +\frac12\cdot\frac m2\\
 &\geq \frac{\OPT_h}{4}+\Omega_k(\eps^3)\OPT_h + \frac{m_\ell}{4} + \frac m4 \\
 &\geq \frac{\OPT_\Phi}{4}+\frac m4
        +\Omega_k(\eps^3)\OPT_h \tag{since $\OPT_h + m_\ell \geq \OPT_\Phi$} \\
 &\geq \frac{\OPT_\Phi}{4}+\frac{\OPT_\Phi}{4}
        +\Omega_k(\eps^3)\OPT_\Phi \tag{since $m/2 \geq \OPT_\Phi/2$ and $\OPT_h \geq 3\OPT_\Phi/16$} \\
 &\geq \left(\frac12+\Omega_k(\eps^3)\right)\OPT_\Phi.
\end{align*}
Off the good event,
\eqref{eq:x1-half-revised} and \eqref{eq:x2-degree-sum-revised} still give
the baseline $\OPT_\Phi/2$. Since the good event has constant probability,
averaging proves the theorem.
\end{proof}

The next subsection studies whether triangle-freeness can also be removed for
Max-$k$XOR.

\subsection{Beyond Triangle-Free Instances: Privacy for Max-$k$XOR with Odd Integer $k$}
\label{sec:dp-kxor-bounded-degree}
In this section, we prove Theorem~\ref{thm:intro-kxor-oddk}. So far, our algorithms for general Max-CSPs have relied on triangle-freeness.
In this subsection we remove that assumption for Max-$k$XOR when $k\geq3$ is
odd.  We first privatize the bounded-influence rounding of
\cite{BarakMORRSTVWW15} and prove a degree-sum guarantee, then combine it
with a private degree partition to remove the maximum-degree bound.  The rounding initially produces an assignment with a large \emph{absolute} advantage; finally, a two-candidate exponential mechanism privately selects its optimal global orientation.

The mechanism is defined on every finite instance, including the empty
instance, and adjacency is addition or removal of one constraint.  For the
utility analysis, let $\Phi$ be a nonempty instance with $m$ constraints
indexed by $\ell\in[m]$.  Each constraint has a scope
$S_\ell\subseteq[n]$ of exactly $k$ distinct variables and a sign
$b_\ell\in\{\pm1\}$; the $\ell$-th predicate is
\[
P_\ell(x_{S_\ell}) \;=\; \frac{1}{2} + \frac{b_\ell}{2}\prod_{i\in S_\ell} x_i\;\in\{0,1\}.
\]
For utility, we assume that the scopes $S_\ell$ are pairwise distinct; the
privacy guarantee itself does not require this restriction.  Write
$d_i=\deg_\Phi(i)$,
$\val_\Phi(x)=\sum_{\ell=1}^mP_\ell(x_{S_\ell})$, and let
\[
    \mathfrak{P}(x)=\frac{\val_\Phi(x)}m-\frac12
    =\frac{1}{2m}\sum_{\ell=1}^m b_\ell \prod_{i\in S_\ell} x_i .
\]



\paragraph{Centered polynomial and normalization.}
Define $g:\{\pm1\}^n\to\mathbb{R}$ by
\[
    g(x)\ :=\ 2\sqrt{m}\,\mathfrak{P}(x)\ =\ \frac{1}{\sqrt{m}}\sum_{\ell=1}^m b_\ell \prod_{i\in S_\ell} x_i.
\]
The multilinear polynomial $g$ has several useful properties for the
analysis.  First, $g$ is homogeneous of degree $k$.  Second, $\E{g}=0$,
and
\begin{align*}
    \Var{g} = \E{g^2}-\E{g}^2=\frac{1}{m}\sum_{\ell'=1}^m\sum_{\ell=1}^mb_{\ell'}b_\ell \;\E{\prod_{i\in S_{\ell'}} x_i\prod_{i\in S_\ell} x_i}
\end{align*}
Recall that on the Boolean cube with uniform measure we have $\E{\prod_{i\in S_{\ell'}} x_i\prod_{i\in S_\ell} x_i}=0$ if $S_{\ell'}\neq S_\ell$ and is 1 otherwise. Hence,
\begin{align*}
    \Var{g} = \frac{1}{m}\sum_{\ell=1}^m b_\ell^2 \;\E{(\prod_{i\in S_{\ell}} x_i)^2}= \frac{1}{m}\sum_{\ell=1}^m 1 = 1
\end{align*}
Moreover, 
\[
\text{Inf}_i[g]
=\sum_{S\ni i}\widehat g(S)^2
=\sum_{\ell:\ i\in S_\ell}\Big(\frac{b_\ell}{\sqrt{m}}\Big)^2
=\frac{1}{m}\sum_{\ell:\ i\in S_\ell}1
=\frac{\deg(i)}{m} \le\ \frac{D}{m}.
\]

Let the level $d$ Fourier weight be $W_d:=\sum_{S:|S|=d}\widehat{g}(S)^2$. Partition the levels: for $s=1,2,\dots,\lceil\log_2 k\rceil$ define 
\[
    I_s=\{d\in\{1,\dots,k\}: 2^{s-1}\leq d\leq 2^s\}
\]
Then there exists $s$ so that 
\[
    \sum_{d\in I_s}W_d\geq \frac{1}{\lceil\log_2 k\rceil}\geq \frac{c}{\log_2 k}.
\]

\subsubsection{The private analogue of the bounded-influence rounding}
We now present a private modification of the constructive rounding underlying the bounded-influence result in \cite{BarakMORRSTVWW15}, in the same three phases:
(i) a data-independent \emph{random restriction} to create first-order mass,
(ii) a \emph{greedy linear boost} implemented privately per coordinate, and
(iii) a data-independent \emph{Chebyshev flip} to harvest degree-$k$ advantage.

\begin{algorithm}[H]
\caption{\textsc{DP-AdvRand}$(\eps_1,\eps_2)$ for odd Max-$k$XOR}
\label{alg:dp-adv}
\begin{algorithmic}[1]

\REQUIRE A possibly empty Max-$k$XOR instance $\Phi$ on
$\X=\{x_1,\ldots,x_n\}$, with odd arity $k$, and privacy parameters
$\epsilon_1,\epsilon_2>0$.\\

\ENSURE Assignment $x_{\mathrm{out}} \in \{\pm1\}^n$.

\STATE \textbf{(Scale selection; data-independent).}
Set $s=\lceil\log_2 k\rceil$, put each index independently into $U$ with
probability $p:=2^{-s}$, and set $F:=[n]\setminus U$.
\STATE \textbf{(Random restriction on $F$; data-independent).}
Draw $x_F\in\{\pm1\}^F$ uniformly at random and write $y:=x_F$.
\STATE \textbf{(Private linear boost on $U$).}
\label{line:odd-k-priv-mech}
For each $j\in U$, let $\mathsf{Act}(j)$ be the set of \emph{active}
constraints whose scope contains $j$ and no other index from $U$.
Define the unnormalized active sum
\[
    L_j:=\sum_{\ell\in\mathsf{Act}(j)}
       b_\ell\prod_{i\in S_\ell\setminus\{j\}}y_i.
\]
Conditional on $(U,y)$, release the signs independently: sample each
$x_j\in\{\pm1\}$ using the two-candidate exponential mechanism with privacy
parameter $\eps_1$ and score $q_j(t):=L_jt$:
\begin{align*}
    \Pr{x_j=+1\mid L_j}
      &=\frac{e^{\eps_1L_j/2}}{e^{\eps_1L_j/2}+e^{-\eps_1L_j/2}}
        =\frac{1+\tanh(\eps_1L_j/2)}2,\\
    \Pr{x_j=-1\mid L_j}
      &=\frac{e^{-\eps_1L_j/2}}{e^{\eps_1L_j/2}+e^{-\eps_1L_j/2}}
        =\frac{1-\tanh(\eps_1L_j/2)}2.
\end{align*}
If $\mathsf{Act}(j)=\varnothing$, then $L_j=0$ and $x_j$ is uniform.
\STATE \textbf{(Chebyshev flip; data-independent).}
\label{line:flip-odd-k}
Pick $r\in\{0,1,\dots,k\}$ uniformly and set $\eta:=\cos(r\pi/k)/2$.
Independently for each $j\in U$, flip $x_j\leftarrow -x_j$ with probability $(1-\eta)/2$.
Let $\hat x=(x_F,x_U)$ be the resulting full assignment.
\STATE \textbf{(Private orientation).}
\label{line:private-orientation}
Use the exponential mechanism with privacy parameter $\eps_2$, candidate set
$\{\hat x,-\hat x\}$, and score $\val_\Phi$ to obtain
$x_{\mathrm{out}}$; that is,
\[
\Pr{x_{\mathrm{out}}=u\mid\hat x}
\ \propto\ \exp\{\eps_2\val_\Phi(u)/2\},
\qquad u\in\{\hat x,-\hat x\}.
\]
\STATE \textbf{(Output).}
Return $x_{\mathrm{out}}$.
\end{algorithmic}
\end{algorithm}

\subsubsection{Privacy}
\begin{theorem}[Privacy of \textsc{DP-AdvRand}]
\label{lem:privacy}
The intermediate assignment $\hat x$ in Algorithm~\ref{alg:dp-adv} is
$(\eps_1,0)$-DP.  The released assignment $x_{\mathrm{out}}$ is
$(\eps_1+\eps_2,0)$-DP.
\end{theorem}

\begin{proof}
Fix the random partition $U$, the restriction $x_F$, and the randomness used
in the Chebyshev flip.  These choices are independent of the instance.  A
constraint is active for exactly one $j\in U$ when its scope meets $U$ only
at $j$; otherwise it is inactive.  Thus adding or removing one constraint
changes at most one active sum $L_j$, and changes it by at most one.  The
score $q_j(t)=L_jt$ consequently has sensitivity one.  The mechanism used
for that coordinate is precisely the $\eps_1$-exponential mechanism, while
all other coordinate mechanisms are unchanged.  Parallel composition,
followed by the data-independent Chebyshev flip, proves that $\hat x$ is
$(\eps_1,0)$-DP.  This argument also covers adjacency between the empty
instance and a one-constraint instance, since no normalization by $m$ is
used by the mechanism.

Conditional on any fixed first-stage output $\hat x$, the candidate set
$\{\hat x,-\hat x\}$ is fixed and adding or removing one constraint changes
the score $\val_\Phi(u)$ by at most one for each candidate.  Hence
Line~\ref{line:private-orientation} is an adaptive
$(\eps_2,0)$-DP exponential mechanism.  Basic adaptive composition (Lemma~\ref{lem:composition}) proves
the claimed $(\eps_1+\eps_2,0)$ privacy guarantee.
\end{proof}

\subsubsection{Utility}
We prove a degree-sum bound, which specializes to the usual $1/\sqrt{D}$ guarantee for bounded-degree graphs. For a fixed partition $U$ and restriction $y = x_F$, let $g_y$ be the restriction of the polynomial $g$ obtained by fixing the variables in $F$. Its linear coefficient corresponding to variable $j\in U$ is exactly
\[
 \Lambda_j := \widehat{g_y}(\{j\}) = \frac{L_j}{\sqrt m}.
\]

\paragraph{Chebyshev flip and coefficient extraction.}
Let $X=(x_j)_{j\in U}$ be the independent signs produced by the private linear boost (Line \ref{line:odd-k-priv-mech}). For a fixed realization of $X$ and $\eta\in[-1/2,1/2]$, let $Z_\eta$ be a random vector with independent coordinates of mean $\eta$, and define
\[
 Q_X(\eta) := \E{g_y(X\cdot Z_\eta)\mid U,y,X}.
\]
This is a univariate polynomial in $\eta$ of degree at most $k$. Its linear coefficient (the coefficient of $\eta$) is exactly
\[
 \sum_{j\in U}\Lambda_j X_j = \frac{1}{\sqrt m}\sum_{j\in U}L_j X_j.
\]
The algorithm evaluates this polynomial implicitly at the $k+1$ Chebyshev nodes $\eta_r=\cos(r\pi/k)/2$. It is a standard fact in approximation theory (see, e.g., \cite{BarakMORRSTVWW15, dinur2007fourier}) that one can extract the linear coefficient of a degree-$k$ polynomial by interpolating its values at these $k+1$ nodes. This interpolation incurs a bounded loss factor of at most $\operatorname{poly}(k)$. Therefore, for a constant $c_k>0$ depending only on $k$, we have
\[
 \frac{1}{k+1}\sum_{r=0}^k |Q_X(\eta_r)| \geq \frac{c_k}{\sqrt m} \left|\sum_{j\in U}L_j X_j\right|.
\]
For each chosen $r$, Jensen's inequality implies $$\E{|g_y(X\cdot Z_{\eta_r})|\mid U,y,X}\geq|Q_X(\eta_r)|.$$ Since the algorithm selects $r$ uniformly at random (Line \ref{line:flip-odd-k}), taking the expectation over the private signs $X$ yields
\begin{align}
\E{|g(\hat x)|\mid U,y}
&\geq \frac{c_k}{\sqrt m}
       \E[X]{\left|\sum_{j\in U}L_j X_j\right|}\notag\\
&\geq \frac{c_k}{\sqrt m}
       \left|\sum_{j\in U}L_j\E{X_j\mid U,y}\right|\notag\\
&=\frac{c_k}{\sqrt m}\sum_{j\in U}|L_j|
       \tanh\left(\frac{\eps_1|L_j|}{2}\right).
\label{eq:exp-linear}
\end{align}
The last equality uses the fact that the two-candidate Exponential Mechanism yields an expected sign of $\E{X_j\mid U,y}=\tanh(\eps_1 L_j/2)$. This coefficient-extraction step is crucial: it ensures that the higher-degree terms of $g_y$ do not interfere with the expectation, even under our biased private signs.

\paragraph{Expected private linear mass.}
Fix $U$ and $j\in U$, and let $a_j := |\mathsf{Act}(j)|$. As a function of the uniform random restriction $y$, we have
\[
 L_j(y) = \sum_{\ell\in\mathsf{Act}(j)} b_\ell\prod_{i\in S_\ell\setminus\{j\}}y_i.
\]
Because the original scopes are pairwise distinct, the terms in this sum are orthogonal. Standard moment bounds for Boolean polynomials of degree $k-1$ give
\begin{equation}\label{eq:level-mass}
 \E[y]{L_j^2}=a_j,
 \qquad\text{and}\qquad
 \E[y]{L_j^4}\leq 9^{k-1}a_j^2.
\end{equation}
This bounded fourth moment allows us to apply the second-moment method (specifically, the Paley-Zygmund inequality~\cite[Section 9.1]{odonnell2014analysis}), which guarantees that $L_j$ deviates significantly from zero with constant probability. In particular, $|L_j|\geq\sqrt{a_j/2}$ occurs with probability at least $1/(4\cdot 9^{k-1})$. Since the function $u\mapsto u\tanh(\eps_1 u/2)$ is monotonically increasing for $u\geq 0$, we can bound its expectation as
\begin{equation}\label{eq:private-linear-mass}
 \E[y]{|L_j|\tanh(\eps_1|L_j|/2)}
 \geq c'_k\sqrt{a_j}\,
       \tanh(c\eps_1\sqrt{a_j}),
\end{equation}
for absolute constants $c, c'_k > 0$.

Let $d_j$ be the degree of variable $j$, and set $a_j=0$ when $j\notin U$. Notice that $0\leq a_j\leq d_j$. Because the function $u\mapsto\tanh(c\eps_1\sqrt u)/\sqrt u$ is non-increasing, we have
\[
 \sqrt{a_j}\tanh(c\eps_1\sqrt{a_j}) \geq \frac{a_j}{\sqrt{d_j}}\tanh(c\eps_1\sqrt{d_j})
\]
whenever $d_j>0$. Furthermore, each constraint incident to $j$ is active with probability $p(1-p)^{k-1}$, meaning $\E[U]{a_j}=p(1-p)^{k-1}d_j$. Averaging equations \eqref{eq:exp-linear} and \eqref{eq:private-linear-mass} over both the partition $U$ and the restriction $y$, and recalling that $\mathfrak P = g/(2\sqrt m)$, we obtain
\begin{equation}\label{eq:val-lower}
 \E{|\mathfrak P(\hat x)|}
 \geq\frac{c_k}{m}\sum_{j:d_j>0}
       \sqrt{d_j}\tanh(c\eps_1\sqrt{d_j}),
\end{equation}
where constants depending only on $k$ are absorbed into $c_k$.

If the maximum degree is bounded by $D$ (i.e., $d_j\leq D$), the same monotonicity implies
\[
 \sqrt{d_j}\tanh(c\eps_1\sqrt{d_j}) \geq \frac{d_j}{\sqrt D}\tanh(c\eps_1\sqrt D).
\]
Since $\sum_j d_j = km$, substituting this into \eqref{eq:val-lower} yields the bounded-degree corollary:
\begin{equation}\label{eq:final-g}
 \E{|\mathfrak P(\hat x)|} \geq \frac{c_k}{\sqrt D}\tanh(c\eps_1\sqrt D).
\end{equation}


\begin{theorem}[Odd-$k$ degree-sum guarantee]
\label{thm:dp-kxor-bounded-degree}
Fix an odd $k\geq3$ and $\eps_1,\eps_2>0$.  There is an absolute constant
$c>0$ and a constant $c_k=\exp(-O(k))$ such that, on every nonempty
Max-$k$XOR instance with pairwise distinct scopes, the intermediate
candidate of
Algorithm~\ref{alg:dp-adv} satisfies
\[
 \E{|\mathfrak P(\hat x)|}
 \geq\frac{c_k}{m}\sum_{j:d_j>0}
       \sqrt{d_j}\tanh(c\eps_1\sqrt{d_j}).
\]
Here $d_j:=\deg_\Phi(j)$.  In particular, if $d_j\leq D$ for all $j$,
for some $D\geq1$, then
\[
 \E{|\mathfrak P(\hat x)|}
 \geq \frac{c_k}{\sqrt D}
       \tanh(c\eps_1\sqrt D).
\]
The final output is $(\eps_1+\eps_2,0)$-differentially private and satisfies
\[
 \E{\mathfrak P(x_{\mathrm{out}})}
 \geq \E{|\mathfrak P(\hat x)|}
       -\frac{2(\log2+1)}{\eps_2m}.
\]
Under the same maximum-degree assumption, this gives
\[
 \E{\mathfrak P(x_{\mathrm{out}})}
 \geq \frac{c_k}{\sqrt D}
       \tanh(c\eps_1\sqrt D)
       -O\left(\frac{1}{\eps_2m}\right).
\]
\end{theorem}

\begin{proof}
The absolute-value bound is \eqref{eq:val-lower}, and privacy is
Theorem~\ref{lem:privacy}.  It remains to analyze the final orientation.
Since every monomial of $\mathfrak P$ has odd degree,
\[
\mathfrak P(-\hat x)=-\mathfrak P(\hat x)
\quad\text{and hence}\quad
\max_{u\in\{\hat x,-\hat x\}}\val_\Phi(u)
=\frac m2+m|\mathfrak P(\hat x)|.
\]
Conditioned on $\hat x$, the expected-utility guarantee for the exponential mechanism in Line~\ref{line:private-orientation}, whose candidate set has size two and whose score has sensitivity one, gives
\[
\E{\val_\Phi(x_{\mathrm{out}})\mid\hat x}
\ge \frac m2+m|\mathfrak P(\hat x)|
-\frac{2(\log2+1)}{\eps_2}.
\]
Subtracting $m/2$, dividing by $m$, and averaging over $\hat x$ proves the
claim.  In fact, the exact conditional advantage is
\[
 \E{\mathfrak P(x_{\mathrm{out}})\mid\hat x}
 =|\mathfrak P(\hat x)|
   \tanh\left(\frac{\eps_2m|\mathfrak P(\hat x)|}{2}\right).
\]
By contrast, choosing $\hat x$ and $-\hat x$ uniformly would give
conditional expected advantage zero.
\end{proof}

\paragraph{Discussion of constants and tightness.}
The dependence $c_k=\exp(-O(k))$ and the $1/\sqrt D$ saturation agree with
the nonprivate bounded-influence analysis of \cite{BarakMORRSTVWW15}.  For
$\eps_1\sqrt D\ll1$, the first-stage advantage is
$\Omega_k(\eps_1)$; once $\eps_1\sqrt D$ is large, it saturates at
$\Omega_k(1/\sqrt D)$.  The private orientation incurs the separate
$O(1/(\eps_2m))$ loss.

\subsubsection{Extension to unbounded degree}
\label{sec:max-kxor-oddk-unbounded}
The degree-sum form of Theorem~\ref{thm:dp-kxor-bounded-degree} lets us use
the degree-partition framework without incorrectly asserting that the
entire low-degree branch has bounded maximum degree.  Consider the following
variant of Algorithm~\ref{alg:main}: set
$\tau=A_k/\eps^2$, spend $\eps/3$ on the noisy degree partition and
$\eps/3$ on the high-degree exponential mechanism, and replace
Line~\ref{line:bounded-deg} by
\textsc{DP-AdvRand}$(\eps/6,\eps/6)$.  As before, return the high-degree
candidate or the \textsc{DP-AdvRand} candidate with equal probability.

\begin{corollary}
\label{cor:max-kxor-oddk-main}
Fix an odd $k\geq3$.  There are constants
$B_k,c_k,\eps_0(k)>0$ such that, for every $0<\eps\leq\eps_0(k)$ and every
Max-$k$XOR instance $\Phi$ with pairwise distinct scopes and
$m\geq B_k/\eps^2$ constraints, the preceding algorithm is
$(\eps,0)$-DP and satisfies
\[
 \E{\val_\Phi(x)}
 \geq \left(\frac12+c_k\eps\right)\OPT_\Phi.
\]
The algorithm does not require triangle-freeness or a maximum-degree bound;
its high-degree branch may take exponential time.
\end{corollary}
\begin{proof}
The three stages use privacy budgets
$\eps/3$, $\eps/3$, and $\eps/6+\eps/6$, respectively.  Adaptive
composition and the final data-independent coin therefore give
$(\eps,0)$-DP.

For utility, let $d_j$ be the true degree of variable $j$, and let $H$ be
the set of constraints all of whose variables have degree at least
$2\tau$.  We use the same two cases as in the proof of
Theorem~\ref{thm:triangle-free}.

\smallskip
\noindent\textbf{Case I: $|H|<m/2$.}
Let $J=\{j:d_j<2\tau\}$.  Every constraint outside $H$ contains a variable
in $J$, so $\sum_{j\in J}d_j\geq m/2$.  For $j\in J$,
$\eps\sqrt{d_j}=O_k(1)$; hence
\[
 \sqrt{d_j}\tanh(c(\eps/6)\sqrt{d_j})
 \geq c'_k\eps d_j.
\]
The degree-sum bound in Theorem~\ref{thm:dp-kxor-bounded-degree}, followed
by the private orientation, shows that the second candidate $x^{(2)}$
satisfies
\[
 \E{\mathfrak P(x^{(2)})}
 \geq\Omega_k(\eps)-O\left(\frac1{\eps m}\right)
 =\Omega_k(\eps).
\]
Here the last equality follows by taking $B_k$ sufficiently large.  The
high-degree candidate
always has expected value at least $m/2$, so the final fair coin gives
$m/2+\Omega_k(\eps)m$ in this case.

\smallskip
\noindent\textbf{Case II: $|H|\geq m/2$.}
The high-degree part of the proof of Theorem~\ref{thm:triangle-free} does
not use triangle-freeness.  With $\tau=A_k/\eps^2$, the same Laplace-tail
and Markov arguments give, with constant probability,
\[
 |\Chigh|\geq\frac{3m}{8},
 \qquad
 |V_h|=O_k(m\eps^2).
\]
For completeness, variables in $V_h$ of true degree greater than
$\tau/2$ number at most $2km/\tau=O_k(m\eps^2)$; the tail bound ensures
that the selected constraints containing any remaining variables number
at most $O(\eps^2m)$, accounting for the second term.  The exponential
mechanism on $V_h$ therefore loses only
$O_k(|V_h|/\eps+1/\eps)=O_k(\eps m)$ constraints.  In the last equality we
also use $m\geq2\tau$, which follows from Case II.  Thus, on this
constant-probability event, the high-degree candidate achieves
\[
 \E{\val_{\Phi_h}(x^{(1)})\mid\Xhigh}
 \geq(1-O_k(\eps))\OPT_h
 \geq\left(\frac12+\Omega_k(\eps)\right)\OPT_h.
\]
Constraints outside $\Chigh$ contribute $1/2$ in expectation.  The
\textsc{DP-AdvRand} candidate also has expected value at least $m/2$,
because its exact oriented advantage is nonnegative.  Combining the two
candidates exactly as in Case II of Theorem~\ref{thm:triangle-free}, using
$\OPT_\Phi\leq\OPT_h+(m-|\Chigh|)$ and
$\OPT_h\geq3m/16$, yields
\[
 \E{\val_\Phi(x)}
 \geq\left(\frac12+\Omega_k(\eps)\right)\OPT_\Phi.
\]
Off the good partition event both candidates retain the $1/2$ baseline, so
averaging preserves the claimed order of advantage.
\end{proof}
\section{Better Approximations on Private Max-Cut}\label{sec:approx-dp-maxcut}


\subsection{Private Max-Cut on Triangle-Free Graphs}

In Section~\ref{sec:csp}, we obtained a $(1/2 + \Omega(\epsilon^3))$-approximation for triangle-free Max-$k$XOR. For Max-Cut, a special case of Max-2XOR, we obtain the sharper $1/2+\Omega(\eps)$ dependence. The improvement comes from privatizing Shearer's specialized triangle-free Max-Cut algorithm~\cite{shearer1992note}. We first analyze the private rounding per edge and then use the degree-partition framework to remove the maximum-degree assumption. The dependence on $\eps$ matches the one-edge hardness in Theorem~\ref{thm:mul_hardness} up to constants and the finite-$n$ term.

\subsubsection{Shearer's Algorithm on Max-Cut}

Throughout the Max-Cut sections, graphs are finite, simple, unweighted, and undirected. For an input graph $G=([n],E)$, Shearer gives the following algorithm to generate a bipartite subgraph, where $d(v)$ is the degree of $v$:

\begin{algorithm}[hbtp]
\caption{A simple random algorithm for Max-Cut~\cite{shearer1992note}}
\label{algo:shearer}

\begin{algorithmic}[1]

\REQUIRE A graph $G = ([n],E)$.\\

\ENSURE A subset of vertices $S\subseteq [n]$.\\

\STATE \textbf{(Initialization).}
For each $v\in [n]$, independently choose two colors $c_1(v), c_2(v)\in\{-1,+1\}$ uniformly at random.

\STATE \textbf{(Assigning $\ell$ values).}
For each node $v\in [n]$, define $\ell(v) = |\{u:\{u,v\}\in E,\ c_1(u) = c_1(v)\}|$.

\STATE \textbf{(Fixing assignment).}
For each $v\in [n]$
\begin{itemize}[label=--]
    \item If $\ell(v) < d(v)/2$, let $c(v) = c_1(v)$.
    \item If $\ell(v) > d(v)/2$, let $c(v) = c_2(v)$.
    \item If $\ell(v) = d(v)/2$, w.p. $1/2$, let $c(v) = c_1(v)$ and w.p. $1/2$, let $c(v) = c_2(v)$.
\end{itemize}

\STATE \textbf{(Output).}
 $S =\{v\in [n]: c(v) = +1\}$.
\end{algorithmic}
\end{algorithm}

Intuitively, the algorithm starts by assigning each vertex a randomly chosen color in $\{+1,-1\}$. Then, each vertex decides whether to resample its color according to the initial colors of its neighbors. Theorem~\ref{thm:shearer} states the formal guarantee of Algorithm~\ref{algo:shearer}.

\begin{theorem}[\cite{shearer1992note}]
\label{thm:shearer}
Let $G = (V,E)$ be a triangle-free graph and let $S$ be the cut found by
Algorithm~\ref{algo:shearer}.  Then, for every edge $e = \{u,v\}\in E$, the
probability that $e$ has exactly one endpoint in $S$ is at least
\[
\frac{1}{2} + \frac{1}{8\sqrt{2d(u)}} + \frac{1}{8\sqrt{2d(v)}}.
\]
\end{theorem}
\begin{remark}
Summing the per-edge guarantee in Theorem~\ref{thm:shearer} gives
\[
\mathbb{E}[\text{val}_G(S)]
\geq \frac{m}{2} + \frac{1}{8\sqrt{2}}\sum_{v\in[n]}\sqrt{d(v)}.
\]
\end{remark}

\subsubsection{Privatizing Shearer's Algorithm on Bounded Degree Instances}
\label{sec:pasin-scratch}

Here, we present a private version of Shearer's algorithm (Algorithm~\ref{algo:shearer}). 
In our privatization (Algorithm~\ref{algo:pure_DP_sha}), let $\ell(v)=|\{u:\{u,v\}\in E,\ c_1(u)=c_1(v)\}|$. We add discrete Laplace noise~\cite{ghosh2009universally} to $\ell(v)-\lceil(d(v)-1)/2\rceil$ and retain the initial color exactly when the perturbed value is nonpositive.
\begin{lemma}\cite{ghosh2009universally}
    \label{lem:laplace-priv-discrete}
    Suppose $f:\mathcal{U}^*\rightarrow \mathbb{Z}^k$ has $\ell_1$ sensitivity $\Delta$. Then the mechanism
    $$\mathcal{M}(D) = f(D) + (Z_1,\cdots,Z_k)^\top$$
    is $(\epsilon,0)$-differentially private, where $Z_1,\ldots,Z_k$ are independent draws from $\DLap(\Delta/\epsilon)$. Here, $\DLap(b)$ assigns mass $\frac{e^{1/b}-1}{e^{1/b}+1}e^{-|z|/b}$ to $z\in\mathbb Z$.
\end{lemma}
\begin{proposition}\label{lem:variance_discrete_laplace}
    Given any $\epsilon > 0$, let $X\sim \text{DLap}(1/\epsilon)$. Then the variance of $X$ is $\text{Var}[X] = \frac{2e^\eps}{(e^\eps - 1)^2}$.
\end{proposition}
\begin{proof}
Since the distribution is symmetric about zero, \(\mathbb{E}[X] = 0\). Thus, \(\text{Var}(X) = \mathbb{E}[X^2]\). We compute:
\[
\mathbb{E}[X^2] = \sum_{x=-\infty}^{\infty} x^2 P(X=x) = 2 \sum_{x=1}^{\infty} x^2 \cdot \frac{e^\epsilon - 1}{e^\epsilon + 1} e^{-\epsilon x}.
\]
Let \(a = e^{-\epsilon}\). Then:
\[
\sum_{x=1}^{\infty} x^2 a^x = \frac{a(1 + a)}{(1 - a)^3}.
\]
Substituting back:
\[
\mathbb{E}[X^2] = 2 \cdot \frac{e^\epsilon - 1}{e^\epsilon + 1} \cdot \frac{e^{-\epsilon}(1 + e^{-\epsilon})}{(1 - e^{-\epsilon})^3}.
\]
Simplifying using \(e^\epsilon - 1 = e^\epsilon(1 - e^{-\epsilon})\) and \(e^\epsilon + 1 = e^\epsilon(1 + e^{-\epsilon})\), we have
\[
\text{Var}(X) = \mathbb{E}[X^2] = \frac{2e^{-\epsilon}}{(1 - e^{-\epsilon})^2} = \frac{2e^{\epsilon}}{(e^{\epsilon} - 1)^2}.
\]
\end{proof}




\begin{algorithm}[h]
\caption{Pure-DP version of Shearer's algorithm}
\label{algo:pure_DP_sha}

\begin{algorithmic}[1]

\REQUIRE A graph $G = ([n],E)$, privacy budget $\epsilon>0$\\

\ENSURE A subset of vertices $S\subseteq [n]$.\\

\STATE \textbf{(Setting parameters).}
Set $\eps_0\leftarrow \eps/2$\;

\STATE \textbf{(Initialization).}
For each $v\in [n]$, independently choose two colors $c_1(v), c_2(v)\in\{-1,+1\}$ uniformly at random.

\STATE \textbf{(Assigning $\ell$ values).}
For each node $v\in[n]$, define $\ell(v)=|\{u:\{u,v\}\in E,\ c_1(u)=c_1(v)\}|$.

\STATE \textbf{(Fixing assignment).}
For each $v\in [n]$
\begin{itemize}[label=--]
    \item Sample $\zeta_v \sim \DLap(1/\eps_0)$.
    \item If $\ell(v) - \left\lceil \frac{d(v)-1}{2} \right\rceil + \zeta_v \leq 0$, let $c(v)\leftarrow c_1(v)$.
    \item Else let $c(v)\leftarrow c_2(v)$.
\end{itemize}

\STATE \textbf{(Output).}
 $S =\{v\in [n]: c(v) = +1\}$.
\end{algorithmic}
\end{algorithm}

Before stating the formal privacy and utility guarantees of Algorithm~\ref{algo:pure_DP_sha}, we need the following technical lemma:
\begin{lemma} \label{lem:at-threshold-prob}
Let $d \in \N$ and $0<\eps \leq 1$. Let $Y = X + Z$ where $X \sim \Bin(d - 1, 1/2)$ and $Z \sim \DLap(1/\eps)$. Then,
\begin{align*}
\PPr\left[Y = \left \lceil \frac{d-1}{2} \right\rceil\right] \geq \Omega\left(\frac{1}{\sqrt{d + 1/\eps^2}}\right).
\end{align*}
\end{lemma}

\begin{proof}
Since $X$ and $Z$ are independent, $Y$ is symmetric around $\mu = \frac{d - 1}{2}$ and has variance $\sigma^2 = O(d + 1/\eps^2)$, which follows from Proposition~\ref{lem:variance_discrete_laplace}. The convolution of two unimodal distributions is unimodal. Thus $Y$ is both symmetric and unimodal, and its maximum point mass is attained at $\left \lceil \frac{d-1}{2} \right\rceil$.

By Chebyshev's inequality, $\PPr[Y \in [\mu - 2\sigma, \mu + 2\sigma]] \geq \frac{3}{4}$. The interval contains at most $4\sigma + 2$ integers. The maximum point mass is at least their average, so

\begin{align*}
\PPr\left[Y = \left \lceil \frac{d-1}{2} \right\rceil\right] \geq \frac{3}{4} \cdot \frac{1}{4\sigma + 2} \geq \Omega\left(\frac{1}{\sqrt{d + 1/\eps^2}}\right).
\end{align*}
\end{proof}

We can now present and prove the guarantees of \Cref{algo:pure_DP_sha}.

\begin{theorem}\label{thm:pure_dp_sha}
     Fix any $\epsilon>0$ and $n\in \mathbb{N}_+$. \Cref{algo:pure_DP_sha} is $\eps$-DP. Furthermore, if $0<\eps\leq 1$ and $G$ is triangle-free, then for any edge $e = \{u,v\}\in E$, the probability that $e$ is cut by the output $S$ of \Cref{algo:pure_DP_sha} is at least $\frac{1}{2} + \Omega\left(\frac{1}{\sqrt{d(u) + 1/\eps^2}} + \frac{1}{\sqrt{d(v) + 1/\eps^2}}\right)$.
\end{theorem}

\begin{proof}
Fix the initial colors $c_1$.  Adding or removing the edge $uv$ changes only the
$u$- and $v$-coordinates of
\[
q(G)=\left(\ell(w)-\left\lceil\frac{d(w)-1}{2}\right\rceil\right)_{w\in V},
\]
and changes each of these coordinates by at most one.  Thus $q$ has
$\ell_1$-sensitivity at most two.  Since $\eps_0=\eps/2$, adding independent
$\DLap(1/\eps_0)=\DLap(2/\eps)$ noise to its coordinates is $\eps$-DP.
The returned cut is a post-processing of this noisy vector and the
data-independent colors, proving privacy.

We next prove utility.  Fix $uv\in E$, write
\[
\tau_w=\left\lceil\frac{d(w)-1}{2}\right\rceil,
\qquad
\mathcal{A}_w=\{\ell(w)+\zeta_w\leq \tau_w\},
\]
so that $\mathcal{A}_w$ is exactly the event that the algorithm uses $c_1(w)$.  Let
$X_1^{uv}$ be $+1$ when $c_1(u)\neq c_1(v)$ and $-1$ otherwise, and define
$X^{uv}$ analogously for the final colors.  Whenever at least one of
$\mathcal{A}_u,\mathcal{A}_v$ fails, at least one endpoint uses a fresh uniform color; hence
the contribution to $\EE[X^{uv}]$ from that event is zero.  Consequently,
\[
\EE[X^{uv}]
=\PPr[X_1^{uv}=1,\mathcal{A}_u,\mathcal{A}_v]
 -\PPr[X_1^{uv}=-1,\mathcal{A}_u,\mathcal{A}_v].
\]

Let $\ell'(w)$ omit the contribution of the edge $uv$ from $\ell(w)$, and set
$Y_w=\ell'(w)+\zeta_w$.  Conditional on either value of $X_1^{uv}$, triangle
freeness implies that $Y_u$ and $Y_v$ are independent, with
$\ell'(w)\sim\Bin(d(w)-1,1/2)$.  Moreover, $\mathcal{A}_w$ is the event
$Y_w\leq\tau_w$ when $X_1^{uv}=1$, and the event
$Y_w\leq\tau_w-1$ when $X_1^{uv}=-1$.  If
$F_w(t)=\PPr[Y_w\leq t]$, then
\begin{align*}
\EE[X^{uv}]
 &=\frac12\big(F_u(\tau_u)F_v(\tau_v)
          -F_u(\tau_u-1)F_v(\tau_v-1)\big)\\
 &\geq \frac12 F_u(\tau_u)
       \big(F_v(\tau_v)-F_v(\tau_v-1)\big)\\
 &\geq \frac14\PPr[Y_v=\tau_v].
\end{align*}
The last inequality uses symmetry of $Y_u$ around $(d(u)-1)/2$, which gives
$F_u(\tau_u)\geq 1/2$.

Assume without loss of generality that $d(u)\geq d(v)$.  Applying
\Cref{lem:at-threshold-prob} with parameter $\eps_0=\eps/2$ yields
\[
\PPr[Y_v=\tau_v]
=\Omega\left(\frac{1}{\sqrt{d(v)+1/\eps^2}}\right).
\]
Since
$1/\sqrt{d(u)+1/\eps^2}\leq1/\sqrt{d(v)+1/\eps^2}$ and
$\PPr[uv\text{ is cut}]=\frac12+\frac12\EE[X^{uv}]$, the claimed
two-endpoint bound follows.
\end{proof}

\subsubsection{Private Max-Cut on Instances with Unbounded Degree}\label{sec:maxcut-unbounded}
We adapt the framework of Section~\ref{sec:csp} to eliminate the maximum-degree
dependence in Theorem~\ref{thm:pure_dp_sha}.  The algorithm privately identifies
vertices that appear to have high degree.  It uses the exponential mechanism on
the subgraph induced by these vertices and assigns every remaining vertex a
uniform random side.  Independently, it runs the private Shearer algorithm on
the full graph, and returns one of the two cuts uniformly at random.
The degree threshold is deliberately chosen with a large constant so that the
exponential-mechanism loss on the high-degree subgraph is small.

\begin{algorithm}[h]
\caption{Private Max-Cut on graphs with unbounded degree}
\label{alg:private_cut}

\begin{algorithmic}[1]

\REQUIRE A graph $G = ([n],E)$, privacy budget $\epsilon>0$.

\ENSURE A subset of vertices $S\subseteq [n]$.\\

\STATE \textbf{(Setting parameters).}
Set $d = \frac{10000}{\epsilon^2}$ and $\mathcal{C} = \varnothing$.

\STATE \textbf{(Private degree estimation).}
For each $v\in [n]$
\begin{itemize}[label=--]
    \item Let $\hat{d}(v) \leftarrow d(v) + \text{Lap}(6/\eps)$
    \item If $\hat{d}(v) > d$, then $\mathcal{C}\leftarrow \mathcal{C} \cup \{v\}$
\end{itemize}

\STATE \textbf{(Separate instances).}
Let $G_1 = (V_1, E_1)$ and $G_2 = ([n], E_2)$ where, 
$$V_1 = \mathcal{C}, E_1 = \{\{u,v\}\in E|\text{ both } u,v\in \mathcal{C}\} \text{ and } E_2 = E\backslash E_1.$$

\STATE \textbf{(Cut on $G_1$).}
Sample a cut $S_1^A \in \{\pm1\}^{|\mathcal{C}|}$ on $G_1$ according to the exponential mechanism:
$$\mathbf{Pr}[S_1^A = \hat{S}] \propto \exp\left( \frac{\epsilon\cdot \text{val}_{G_1}(\hat S)}{6}\right).$$

\STATE \textbf{(Cut on $[n]\setminus\mathcal{C}$).}
Let $S_1^B\in \{\pm1\}^{[n]\setminus \mathcal{C}}$ be a uniformly random cut on vertices $[n]\setminus \mathcal{C}$.

\STATE \textbf{(Global cut $S_1$).}
Let $S_1\in \{\pm1\}^{[n]}$ be the concatenation of $S_1^A$ and $S_1^B$ in the natural way.

\STATE \textbf{(Global cut $S_2$).}
Let $S_2\in \{\pm1\}^{[n]}$ be the output of Algorithm~\ref{algo:pure_DP_sha} on $G$ with parameter $\epsilon/3$.

\STATE \textbf{(Output).}
With probability $\frac{1}{2} $, output $S = S_1$ and 
with probability $\frac{1}{2} $, output $S = S_2$. 
\end{algorithmic}
\end{algorithm}

\begin{theorem}\label{thm:privacy_private_maxcut}
    For every $\epsilon>0$, Algorithm~\ref{alg:private_cut} is $(\epsilon,0)$-differentially private.
\end{theorem}
\begin{proof}
The degree vector has $\ell_1$-sensitivity two under edge add/remove
adjacency.  Thus adding independent $\mathrm{Lap}(6/\eps)$ noise releases the private
partition with privacy parameter $\eps/3$.  Conditional on any fixed value of
$\cC$, the score $\val_{G[\cC]}(S)$ has sensitivity one and the coefficient
$\eps/6$ in the exponential mechanism corresponds to privacy parameter
$\eps/3$.  Finally, Algorithm~\ref{algo:pure_DP_sha}, run with parameter
$\eps/3$, is $\eps/3$-DP.  The uniform labels and final coin toss are
post-processing (Lemma~\ref{lem:postprocess}).  Adaptive composition (Lemma~\ref{lem:composition}) therefore gives total privacy
$\eps/3+\eps/3+\eps/3=\eps$.
\end{proof}

\begin{theorem}\label{thm:utility_maxcut_unbounded}
There is a universal constant $\eps_0>0$ such that, for every
$0 < \epsilon \leq \eps_0$ and every triangle-free graph on $n$ vertices with
$m$ edges, Algorithm~\ref{alg:private_cut} outputs a random cut $S\subseteq[n]$
such that
     $$ \mathbb{E}[\text{val}_G(S)] \geq \left(\frac{1}{2} +\Omega\left(\epsilon\right)\right) \text{OPT}_{G}.$$
    Here, $\text{val}_G(S)$ is the size of cut $S$ with respect to $G$, and $\text{OPT}_{G}$ is the size of max-cut in $G$.
\end{theorem}
\begin{proof}
Write $d=10000/\eps^2$.  Conditional on any fixed $\cC$, every edge not
in $E_1=E(G[\cC])$ has a uniformly labeled endpoint and is therefore cut by
$S_1$ with probability $1/2$.  On $G_1$, exponentially tilting the uniform
distribution by the nonnegative score $\val_{G_1}$ cannot decrease its
expected score.  Thus
\begin{equation}\label{eq:g_1_better_than_half}
 \E{\val_G(S_1)\mid\cC}\geq m/2.
\end{equation}
Theorem~\ref{thm:pure_dp_sha} likewise gives
\begin{equation}\label{eq:g_2_better_than_half}
 \E{\val_G(S_2)}\geq m/2.
\end{equation}

Let $H$ be the set of edges both of whose endpoints have degree at least
$10d$.  If $|H|\leq m/2$, every edge in $E\setminus H$ has an endpoint of
degree less than $10d$.  Applying Theorem~\ref{thm:pure_dp_sha} with privacy
parameter $\eps/3$ gives
\[
 \E{\val_G(S_2)}
 \geq \frac m2+\Omega\left(
   \sum_{uv\in E}\left(
    \frac1{\sqrt{d(u)+9/\eps^2}}+
    \frac1{\sqrt{d(v)+9/\eps^2}}
   \right)\right)
 \geq \frac m2+\Omega(\eps m).
\]
Together with \eqref{eq:g_1_better_than_half} and the final fair coin, this
proves the result in this case.

It remains to consider $|H|>m/2$.  Let $L$ be the set of edges with an
endpoint of degree less than $d/10$.  The exact Laplace tail implies
\begin{align*}
 \PPr[v\notin\cC]&\leq \tfrac12e^{-15000/\eps}
       &&\text{if }d(v)\geq10d,\\
 \PPr[v\in\cC]&\leq \tfrac12e^{-1500/\eps}
       &&\text{if }d(v)<d/10.
\end{align*}
Consequently,
\[
 \E{|H\setminus E_1|}\leq me^{-15000/\eps},
 \qquad
 \E{|L\cap E_1|}\leq \tfrac12me^{-1500/\eps}.
\]
Applying Markov's inequality to the two nonnegative numbers of bad edges
shows that, with probability at least $4/5$ (for $\eps_0$ chosen as a small
universal constant), the event
\begin{equation}\label{eq:good-triangle-partition}
 |H\setminus E_1|\leq m/10
 \quad\text{and}\quad
 |L\cap E_1|\leq m\eps^2/100
\end{equation}
holds.  Notice that Markov is used here on the number of errors, rather than
to infer a lower tail from the expected number of retained edges.

Fix a partition satisfying \eqref{eq:good-triangle-partition}.  Then
$m_1=|E_1|\geq2m/5$.  Moreover, the number $r$ of non-isolated vertices of
$G_1$ is at most
\[
 r\leq 2|L\cap E_1|+|\{v:d(v)\geq d/10\}|
 \leq \frac{m\eps^2}{50}+\frac{20m}{d}
 \leq 0.022m\eps^2.
\]
The score marginal of the exponential mechanism is unchanged if we restrict
it to these $r$ vertices, since every active assignment has the same number of
extensions to the isolated vertices.  Its expected-utility guarantee gives
\[
 \E{\val_{G_1}(S_1^A)\mid\cC}
 \geq \OPT_{G_1}-\frac6\eps(r\log2+1).
\]
The present case is nonempty only if $m\geq10d$, and hence the final
$6/\eps$ term is also $O(m\eps)$. By decreasing $\eps_0$ if necessary, we can bound this additive loss relative to the optimum: since $\OPT_{G_1} \geq m/5$, we have $O(m\eps) = O(\eps\OPT_{G_1})$. Thus, the expected-utility guarantee on $G_1$ becomes:
\[
 \E{\val_{G_1}(S_1^A)\mid\cC} \geq \OPT_{G_1} - O(\eps\OPT_{G_1}) = (1 - O(\eps))\OPT_{G_1}.
\]
For sufficiently small $\eps \le \eps_0$, this expectation is lower-bounded by $\left(\frac{1}{2} + \Omega(\eps)\right)\OPT_{G_1}$.

Therefore, letting $m_2 = m - m_1$, conditional on the good partition event, we combine the contributions of $S_1$ and $S_2$ over the final fair coin. Notice that $S_1$ cuts exactly $m_2/2$ edges in expectation outside $E_1$. Putting everything together, and by the fact that $E_1 \cap E_2 = \varnothing$, we have 
\begin{align*}
 \E{\val_G(S)\mid\cC} 
 &= \frac{1}{2}\E{\val_G(S_1)\mid\cC} + \frac{1}{2}\E{\val_G(S_2)\mid\cC} \\
 &= \frac{1}{2} \left( \E{\val_{G_1}(S_1^A)\mid\cC} + \frac{m_2}{2} \right) + \frac{1}{2}\E{\val_G(S_2)\mid\cC} \\
 &\geq \frac{1}{2}\left( \left(\frac{1}{2} + \Omega(\eps)\right)\OPT_{G_1} + \frac{m_2}{2} \right) + \frac{1}{2}\cdot\frac{m}{2} \tag{from \eqref{eq:g_2_better_than_half}} \\
 &= \frac{\OPT_{G_1}}{4} + \Omega(\eps\OPT_{G_1}) + \frac{m_2}{4} + \frac{m}{4} \\
 &\geq \frac{\OPT_G}{4} + \frac{m}{4} + \Omega(\eps\OPT_{G_1}) \tag{since $\OPT_{G_1} + m_2 \geq \OPT_G$} \\
 &\geq \frac{\OPT_G}{4} + \frac{m}{4} + \Omega(\eps\OPT_G) \tag{since $m/2 \geq \OPT_G/2$ and $\OPT_{G_1} \geq m/5 \geq \OPT_G/5$} \\
 &\geq \left(\frac{1}{2} + \Omega(\eps)\right)\OPT_G.
\end{align*}
Off the good event, \eqref{eq:g_1_better_than_half} and \eqref{eq:g_2_better_than_half} still guarantee the baseline $\E{\val_G(S)\mid\cC} \geq m/2 \geq \OPT_G/2$. Since the good event holds with probability at least $4/5$, averaging over the private partition yields:
\[
 \E{\val_G(S)} \geq \frac{4}{5}\left(\frac{1}{2} + \Omega(\eps)\right)\OPT_G + \frac{1}{5}\left(\frac{1}{2}\right)\OPT_G = \left(\frac{1}{2} + \Omega(\eps)\right)\OPT_G.
\]
\end{proof}

\subsection{Private Max-Cut on General Graphs}\label{sec:max-cut_without_trianglefreeness}

The preceding triangle-free result has optimal-order dependence on $\eps$. We now remove triangle-freeness at the cost of a logarithmic factor and a density requirement.

\begin{algorithm}[H]
\caption{Private Max-Cut on general graphs}
\label{alg:cut_general_graph}

\begin{algorithmic}[1]

\REQUIRE A graph $G = ([n],E)$, privacy budget $\epsilon>0$, and parameter $\alpha>0$.

\ENSURE A subset of vertices $S\subseteq [n]$.\\

\STATE \textbf{(Setting parameters).}
Set $d = \frac{24}{\epsilon^{1+\alpha}}$ and $V_1 = \varnothing$.

\STATE \textbf{(Private degree estimation).}
For each $v\in [n]$
\begin{itemize}[label=--]
    \item Let $\hat{d}(v) \leftarrow d(v) + \text{Lap}(12/\eps)$
    \item If $\hat{d}(v) > d$, then $V_1\leftarrow V_1 \cup \{v\}$
\end{itemize}

\STATE \textbf{(High degree instance).}
Let $G_1 = (V_1, E_1)$ where, 
$E_1 = \{\{u,v\}\in E|\text{ both } u,v\in V_1\}.$

\STATE \textbf{(Cut on $G_1$).}
Sample a cut $S_1^A \in \{\pm1\}^{|V_1|}$ on $G_1$ according to the exponential mechanism:
$$\mathbf{Pr}[S_1^A = \hat{S}] \propto \exp\left( \frac{\epsilon\cdot \text{val}_{G_1}(\hat S)}{12}\right).$$

\STATE \textbf{(Low degree instance).}
\begin{itemize}[label=--]
\item Let $\widetilde{G} = (V_2,\widetilde{E})$ where, $V_2 = V\setminus V_1$ and $\widetilde{E} = \{\{u,v\}\in E|\text{ both } u,v\in V_2\}$.
\item $E_2 \leftarrow \emptyset$. Sample each edge in $\widetilde{E}$ independently with probability $\eps^{1+\alpha} / 70$ and add it to $E_2$.
\item Set $E_M \gets \emptyset$. Let $G_2 = (V_2, E_2)$ and let $N_2(v)$ be the neighborhood of $v$ in $G_2$.
\item  For each $v \in V_2$, let $f(v)$ be a uniformly random neighbor from $N_2(v)$. If $v$ is isolated, let $f(v) = \bot$.
\item For each $(u,v) \in E_2$, if $f(v) = u$ and $f(u)=v$, $E_M \leftarrow E_M \cup \{(u,v)\}$.
\item Let $G_M = (V_2, E_M)$.
\end{itemize}

\STATE \textbf{(Cut on $G_M$).}
Sample a cut $S_2^A \in \{\pm1\}^{|V_2|}$ on $G_M$ according to the exponential mechanism:
$$\mathbf{Pr}[S_2^A = \hat{S}] \propto \exp\left( \frac{2.5\cdot \text{val}_{G_M}(\hat S)}{4}\right).$$

\STATE \textbf{(Cut on $[n]\setminus V_1$).}
Let $S_1^B\in \{\pm1\}^{[n]\setminus V_1}$ be a uniformly random cut on vertices $[n]\setminus V_1$.

\STATE \textbf{(Cut on $[n]\setminus V_2$).}
Let $S_2^B\in \{\pm1\}^{[n]\setminus V_2}$ be a uniformly random cut on vertices $[n]\setminus V_2$.

\STATE \textbf{(Global cut $S_1$).}
Let $S_1\in \{\pm1\}^{[n]}$ be the concatenation of $S_1^A$ and $S_1^B$ in the natural way.

\STATE \textbf{(Global cut $S_2$).}
Let $S_2\in \{\pm1\}^{[n]}$ be the concatenation of $S_2^A$ and $S_2^B$ in the natural way.

\STATE \textbf{(Global cut $S_3$).}
Let $S_3 = (V_1, V_2)$.

\STATE \textbf{(Output).}
Use Exponential Mechanism with budget $\frac{\eps}{2}$ to output the best cut among $\{S_1, S_2, S_3\}$. 

\end{algorithmic}
\end{algorithm}
\begin{theorem}\label{thm:maxcut_without_trianglefreeness}
There are universal constants $c_0,C_0>0$ such that the following holds.
Let $\alpha>0$, $0<\eps\leq0.1$, and put $q=\eps^{1+\alpha}$.
Suppose
\[
 \eps^\alpha\leq
 \min\left\{c_0,\frac{1.8}{\ln(10/q)}\right\}
 \qquad\text{and}\qquad
 m\geq \frac{C_0}{\eps q}.
\]
Then Algorithm~\ref{alg:cut_general_graph} is $(\eps,0)$-differentially
private and outputs a cut $S$ satisfying
\[
 \E{\val_G(S)}\geq
 \left(\frac12+\Omega(q)\right)\OPT_G.
\]
\end{theorem}
\begin{remark}
For sufficiently small fixed $c>0$, setting
$\alpha=\log_{1/\eps}(\ln(1/\eps)/c)$ gives
$q=\Theta(\eps/\ln(1/\eps))$ and satisfies the first condition for all
sufficiently small $\eps$.  The density condition then becomes
$m=\Omega(\ln(1/\eps)/\eps^2)$.  The logarithmic density factor is needed in
this proof to absorb the $O(1/\eps)$ loss of the final private selection. This directly implies Theorem~\ref{thm:intro_maxcut_trianglefreeness}.
\end{remark}

\begin{proof}
We first prove privacy under edge add/remove adjacency.  The degree vector has
$\ell_1$-sensitivity two, so noise of scale $12/\eps$ makes the partition
$\eps/6$-DP.  Conditional on any fixed partition, the exponential mechanism
on $G_1$ has score sensitivity one and privacy parameter $\eps/6$.

We next justify the privacy of the low-degree branch.  To analyze the base
mechanism to which amplification is applied, consider neighboring sampled
graphs $H$ and $H'=H+ab$.  Couple the random-neighbor
choices as follows.  Choices outside $a,b$ are identical.  For $a$, first
sample its choice uniformly from $N_{H'}(a)$.  If it is not $b$, use the same
choice in $H$; otherwise sample a uniform element of $N_H(a)$ (or $\bot$ if
that set is empty).  This gives the correct uniform marginal in both graphs;
couple the choice at $b$ in the same way.  Under this coupling, the two
matchings can differ only by deleting at most one old matching edge at each of
$a,b$ and possibly adding $ab$.  Hence, for every cut $T$,
\begin{equation}\label{eq:matching-coupling-sensitivity}
 |\val_M(T)-\val_{M'}(T)|\leq2.
\end{equation}
For each coupled pair $(M,M')$, the distributions with density proportional
to $\exp((2.5/4)\val_M(T))$ are therefore within an $e^{2.5}$ likelihood
ratio: one factor $e^{1.25}$ accounts for the numerator and one for the
normalizing constant.  Averaging over the coupling proves that the entire
random-matching mechanism is $(2.5,0)$-DP. By \Cref{lem:privacy_amplification}, Poisson subsampling each edge with probability $p = q/70$ amplifies the privacy parameter of this base mechanism to
\[
 \ln\left(1+\frac q{70}(e^{2.5}-1)\right)
 \leq \frac q6 \leq \frac \eps 6,
\]
where the inequality follows from $q\leq\eps$ and $e^{2.5}-1<70/6$. The final exponential mechanism over $\{S_1, S_2, S_3\}$ uses a privacy parameter of $\eps/2$, and the independent labels are post-processing (Lemma~\ref{lem:postprocess}). Therefore, adaptive composition (\Cref{lem:composition}) gives a total privacy loss of
\[
 \frac\eps6+\frac\eps6+\frac\eps6+\frac\eps2=\eps.
\] 
We turn to utility.  Put $d=24/q$, and let $L$ be the set of edges with an
endpoint of degree less than $d/10$.  Let $H$ be the set of edges with an
endpoint of degree at least $2d$.  The exact Laplace tail and the hypothesis
on $\eps^\alpha$ give
\begin{align*}
 \E{|L\cap E_1|}
 &\leq \tfrac12e^{-1.8/\eps^\alpha}m\leq0.05qm,\\
 \E{|H\cap\widetilde E|}
 &\leq \tfrac12e^{-2/\eps^\alpha}m\leq0.05qm.
\end{align*}
Markov's inequality and a union bound show that, with probability at least
$1/2$ over the private partition, the event
\begin{equation}\label{eq:good-general-partition}
 |L\cap E_1|\leq0.2qm
 \quad\text{and}\quad
 |H\cap\widetilde E|\leq0.2qm
\end{equation}
holds.  We condition temporarily on such a partition.  Let $E_\times$ be the
edges between $V_1$ and $V_2$.  The three sets $E_1,\widetilde E,E_\times$
partition $E$, so one of the following cases applies.

\paragraph{(a) Many crossing edges.}
If $|E_\times|\geq0.6m$, then $S_3=(V_1,V_2)$ cuts at least $0.6m$, and hence
has value at least $(1/2+0.1)\OPT_G$.

\paragraph{(b) Many edges inside $V_1$.}
Suppose $|E_1|\geq0.2m$.  Every non-isolated vertex of $G_1$ either has degree
at least $d/10$ in $G$, or is an endpoint of an edge in $L\cap E_1$.
Therefore the number $r_1$ of non-isolated vertices of $G_1$ satisfies
\[
 r_1\leq2|L\cap E_1|+\frac{20m}{d}
 \leq1.24qm.
\]
The exponential mechanism on $G_1$, whose privacy parameter is $\eps/6$,
has the same score marginal as the mechanism restricted to these $r_1$
vertices, because every assignment to them has the same number of extensions
to isolated vertices.  Its expected value is therefore at least
\[
 \OPT_{G_1}-\frac{12}{\eps}(r_1\log2+1).
\]
The good event and $q\leq c_0\eps$ imply that $E_1\setminus L$ is nonempty
(unless $m=0$, when the theorem is trivial), and hence $m\geq d/10$.
Thus the displayed loss is at most $O(\eps^\alpha m)$.  Choosing the universal
constant $c_0$ sufficiently small makes it at most $\OPT_{G_1}/4$, because
$\OPT_{G_1}\geq|E_1|/2\geq0.1m$.  All edges outside $E_1$ have a uniformly
labeled endpoint in $S_1$, so
\begin{align*}
 \E{\val_G(S_1)\mid V_1}
 &\geq\frac34\OPT_{G_1}+\frac12(m-|E_1|)\\
 &\geq\frac12\OPT_G+\frac14\OPT_{G_1}
 \geq\left(\frac12+\frac1{40}\right)\OPT_G.
\end{align*}

\paragraph{(c) Many edges inside $V_2$.}
If neither previous case applies, then $|\widetilde E|>0.2m$.
By \eqref{eq:good-general-partition} and $q\leq0.1$, at least $0.18m$ edges
of $\widetilde E$ have both endpoints of degree less than $2d$ in $G$.
Fix one such edge $e=uv$. Conditional on sampling $e$, let $D_u,D_v$ be the
degrees of its endpoints in the sampled graph $G_2$. The remaining incident-edge
samples at $u$ and $v$ are independent, and
\[
 \E{D_u},\E{D_v}\leq1+\frac{2dq}{70}=1+\frac{48}{70}.
\]
By independence and Jensen's inequality,
\[
 \PPr[e\in M]
 =\frac q{70}\E{\frac1{D_uD_v}}
 =\frac q{70}\E{\frac1{D_u}}\E{\frac1{D_v}}
 \geq\frac q{70}\frac{1}{\E{D_u}}\frac{1}{\E{D_v}}
 \geq\frac q{70(1+48/70)^2}=\Omega(q).
\]
Summing over the $0.18m$ eligible edges gives
$\E{|M|\mid V_1}=\Omega(qm)$.

For a fixed matching $M$, the exponential-mechanism distribution factorizes
over its edges and isolated vertices. For the assignment $S_2$, we have:
\begin{itemize}[label=-]
    \item Edges with at least one endpoint in $V_1$ are cut with probability $1/2$. This is because vertices in $V_1$ are assigned randomly.
    \item Edges with at least one endpoint being an isolated vertex in $G_M$ are cut with probability $1/2$. Any assignment to such vertices does not contribute to the score, so the exponential mechanism is equally likely to assign it $+1$ or $-1$.
    \item Edges connecting two matching edges are cut with probability $1/2$. Consider an assignment from the exponential mechanism to vertices incident on a fixed matching edge. Flipping the assignment of these vertices does not change whether the fixed matching edge is cut (and hence the score is unchanged). Because each assignment and its flip are equally likely, the edge incident on two matching edges is cut with probability $1/2$.
\end{itemize}

Each matching edge is cut with probability
\[
 p=\frac{e^{2.5/4}}{1+e^{2.5/4}}>0.65.
\]
As a result, the expected value of cut $S_2$ in this case is
\begin{align*}
 \E{\val_G(S_2)\mid V_1}
 &= \frac{1}{2}\left( m - \E{|M|\mid V_1}\right) + p \E{|M|\mid V_1} \\
 &= \frac m2+\left(p-\frac12\right)\E{|M|\mid V_1} \\
 &\geq\frac m2+\Omega(qm) \\
 &\geq\frac12\OPT_G+\Omega(q\OPT_G).
\end{align*}

For every partition, even one outside
\eqref{eq:good-general-partition}, the conditional expected value of $S_1$
is at least $m/2$: exponentially tilting the uniform distribution on cuts of
$G_1$ cannot decrease its mean, and every remaining edge has a uniform
endpoint. Hence the expected maximum value among $S_1,S_2,S_3$ is always at
least $\OPT_G/2$. On the good event, which has probability at least $1/2$,
the preceding case analysis improves this by $\Omega(q\OPT_G)$. It follows
that
\[
 \E{\max_{i\in\{1,2,3\}}\val_G(S_i)}
 \geq\frac12\OPT_G+\Omega(q\OPT_G).
\]
Finally, we use the exponential mechanism to output the best cut among the three. The mechanism has score sensitivity one, candidate set of size three, and privacy parameter $\eps/2$, so its expected error is at most $4(\ln3+1)/\eps$. Since $\OPT_G\geq m/2$ and $m\geq C_0/(\eps q)$, taking $C_0$ sufficiently large absorbs this additive error into the $\Omega(q\OPT_G)$ term and completes the proof of Theorem~\ref{thm:maxcut_without_trianglefreeness}.
\end{proof}

\section{Hardness Results}\label{sec:additive-lower-bound}

In this section, we introduce several hardness results in terms of the approximation of Max-CSPs with both multiplicative error and additive error.

\subsection{Multiplicative Upper Bounds}\label{sec:mul_hardness}
We first prove a hardness result for a fixed predicate family closed under all input-negation patterns. For $\epsilon<1$, it bounds the value of every $(\epsilon,\delta)$-DP algorithm by $\mu(P)+O(\epsilon)+\delta$ on some one-constraint instance from that family. This does not cover Max-Cut, whose edge predicates do not include arbitrary input-negation patterns, so we give a separate argument for Max-Cut.
\subsubsection{An upper bound for general $k$-CSPs}
Before presenting the hardness result, we define precisely the family to which it applies. Fix $n,k\geq 1$ and a publicly known predicate $P: \{-1,+1\}^k \rightarrow \{0,1\}$. An instance over variables $x_1,\ldots,x_n$ is a multiset of $P$-constraints. Each constraint is a pair $(c,S)$, where $S\in [n]^k$ is an ordered scope with distinct entries and $c\in \{-1,1\}^k$ is a negation pattern. Let $\mu(P)$ be the acceptance probability of $P$ under a uniformly random input. Thus, the result below is a lower bound for every fixed-predicate family closed under all input-negation patterns; it is not a statement about an arbitrary heterogeneous collection of predicates.
\begin{theorem}\label{thm:mul_hardness_csp}
    Fix a nonconstant predicate $P$, $\eps\geq 0$, and $0\leq \delta<1$. If $M$ is $(\epsilon,\delta)$-DP, then there exists a one-constraint CSP instance $\Phi$ with respect to $P$ such that
    $$\mathbb{E}[\text{val}_{\Phi}(M(\Phi))] \leq 1+\delta - e^{-\epsilon}( 1- \mu(P)).$$
    In particular, when $\eps \le 1$ and $\delta = 0$, we have $\mathbb{E}[\text{val}_{\Phi}(M(\Phi))] \leq \mu(P) + (1-\mu(P))\epsilon$; when $\eps \gg 1$ and $\delta=0$, the displayed bound is $1-(1-\mu(P))e^{-\epsilon}$.
\end{theorem}
\begin{proof}
    Let $M$ be a mechanism that outputs an assignment on every $P$-instance. Fix a scope $S$ of $k$ variables. For a negation pattern $c$, let $U_c\subseteq\{-1,+1\}^k$ be the assignments on $S$ satisfying $(c,S)$. Let $\mathcal D$ be the distribution of $M(\varnothing)$ restricted to $S$. Averaging over all negation patterns shows that some $c^*$ satisfies $\mathcal D(U_{c^*})\leq\mu(P)$. Let $\Phi$ contain only the constraint $(c^*,S)$. Then
    \begin{align*}
    \mathbf{Pr}[\text{val}_{\Phi}(M(\Phi)) = 0] &\geq e^{-\epsilon}\mathbf{Pr}[M(\varnothing) \text{ does not fall into } U_{c^*}] - \delta \\
    &\geq e^{-\epsilon} (1-\mu(P)) - \delta.
    \end{align*}
    It implies that 
    $$1- \mathbf{Pr}[\text{val}_{\Phi}(M(\Phi)) = 1] \geq e^{-\epsilon}(1-\mu(P)) - \delta.$$
    Namely,
    $$\mathbf{Pr}[\text{val}_{\Phi}(M(\Phi)) = 1] \leq e^{-\epsilon}\mu(P) + 1-e^{-\epsilon} + \delta.$$
\end{proof}

\subsubsection{An upper bound for Max-Cut}

\begin{theorem}\label{thm:mul_hardness}
    Fix $n\geq2$ and $\epsilon\geq0$. If $M$ is an $(\epsilon,0)$-DP algorithm that outputs a cut on every $n$-vertex graph, then there is a one-edge graph $\Phi^*$ for which
    \[
    \mathbb{E}[\text{val}_{\Phi^*}(M(\Phi^*))]
    \leq \left(\frac{e^\epsilon}{2}+\frac{e^\epsilon}{2(n-1)}\right)\OPT_{\Phi^*}.
    \]
\end{theorem}
\begin{proof}
    Let $\tilde{\Phi}$ contain one uniformly random edge $\{i,j\}$. Given a fixed assignment $\mathbf{x}\in\{-1,+1\}^n$, let $S_{\mathbf{x}}\subseteq[n]$ be the indices assigned $+1$. The edge is cut exactly when one of $i,j$ belongs to $S_{\mathbf{x}}$. Differential privacy between the empty graph and a one-edge graph gives
    \begin{equation*}
        \begin{aligned}
            \mathbf{Pr}_{\tilde{\Phi}, M}[\text{val}_{\tilde\Phi}(M(\tilde\Phi)) = 1] &\leq e^\epsilon \cdot  \mathbf{Pr}_{\tilde{\Phi}, M}[\text{val}_{\tilde\Phi}(M(\varnothing)) = 1]\\
            & = e^\epsilon \cdot \sum_{\mathbf{x}\in \{-1,+1\}^n}\mathbf{Pr}[{M(\varnothing) = \mathbf{x}}] \cdot\mathbf{Pr}_{\{i,j\}\sim\operatorname{Unif}(\binom{[n]}{2})} \left[|\{i,j\}\cap S_{\mathbf{x}}|=1\right]\\
            & = e^\epsilon \cdot \sum_{\mathbf{x}\in \{-1,+1\}^n}\mathbf{Pr}[{M(\varnothing) = \mathbf{x}}] \cdot \frac{|S_\mathbf{x}| (n-|S_\mathbf{x}|)}{\binom n2}\\
            &\leq e^\epsilon \cdot \left( \frac{1}{2} + \frac{1}{2(n-1)}\right)
        \end{aligned}
    \end{equation*}
    for $\epsilon \geq 0$. The second equality holds because an edge $\{i,j\}$ is cut exactly when one of $i,j$ is in $S_\mathbf{x}$. Thus there is a particular one-edge Max-Cut instance $\Phi^*$ for which the expected approximation ratio is at most $\frac{e^\epsilon}{2}+\frac{e^\epsilon}{2(n-1)}$.
\end{proof}

\subsection{A weighted $\ell_1$ additive lower bound}
\label{sec:weighted-l1-lb}
The following result concerns a hardness result for Max-Cut in a separate metric model for \textit{weighted} graphs. A
weighted graph is represented by a nonnegative vector
$w\in\mathbb R_{\geq0}^{\binom n2}$, with one coordinate for each unordered
pair of vertices. Two weighted graphs are adjacent when
$\|w-w'\|_1\leq1$. This is the standard \emph{edge}-level differential privacy (e.g., \cite{dwork2006calibrating, gupta2012iterative, eliavs2020differentially, liu2024optimal, aamand2025breaking}) and is a natural extension of the unweighted edge-add/remove model used in the algorithmic sections to weighted settings. For a cut $R\subseteq[n]$, write
\[
 \Phi_G(R)=\sum_{1\leq i<j\leq n}
 w_{ij}\,\mathbf{1}\bigl\{|\{i,j\}\cap R|=1\bigr\}.
\]
For every fixed $R$, changing the weight vector by $\ell_1$ distance at most
one changes $\Phi_G(R)$ by at most one. Thus cut value has sensitivity one in
this model.

\begin{theorem}\label{thm:lower_bound}
For every even $n\geq512$, $0<\epsilon\leq n/256$ and every $(\epsilon,0)$-DP algorithm $\mathcal{M}$ for weighted-$\ell_1$ adjacency, there exists a graph $G$ such that for $R \sim \mathcal{M}(G)$:
 \[
  \mathbb E[\OPT_G-\Phi_G(R)] \geq \frac{1}{4096}\cdot \frac{n}{\epsilon}.
\]
Here, $\OPT_G$ is the size of the max-cut on $G$.
\end{theorem}
We note that this $\Omega(n)$ lower bound is matched by the exponential mechanism due to Lemma~\ref{prop:exponential_mechanism}, as the candidate space has size $2^{n-1}$ and the score sensitivity is one.

\begin{lemma}\label{lem:exist_large_set}
For every even $n\geq512$, there is a family $\mathcal A$ of balanced subsets of $[n]$ such that
\[
 |\mathcal A|\geq \tfrac13 e^{n/64}
 \quad\text{and}\quad
 \frac n8<|S\cap T|<\frac{3n}{8}
 \quad(S\ne T\in\mathcal A).
\]
\end{lemma}
\begin{proof}
Let $S$ and $T$ be independent uniformly random subsets of $[n]$ of size
$n/2$. After fixing $S$, the random variable
\[
 X:=|S\cap T|
\]
is hypergeometric: the population has size $n$, exactly $n/2$ elements are in
$S$, and $T$ consists of $n/2$ samples without replacement. In particular,
\[
 \E{X}=\frac n2\cdot\frac{n/2}{n}=\frac n4.
\]
Hoeffding's inequality for sampling without replacement gives
\[
 \Pr{\left|X-\frac n4\right|\geq\frac n8}
 \leq 2\exp\left(-\frac{2(n/8)^2}{n/2}\right)
 =2e^{-n/16};
\]
see, for example, \cite{dubhashi2009concentration}.

Now choose $N=\lfloor\tfrac12e^{n/64}\rfloor$ balanced sets
$S_1,\ldots,S_N$ independently and uniformly. For each pair $a<b$, call the
pair bad if
\[
 |S_a\cap S_b|\notin(n/8,3n/8).
\]
The preceding estimate and a union bound over the $\binom N2$ pairs imply
\begin{align*}
 \Pr{\text{some pair is bad}}
 &\leq \binom N2\,2e^{-n/16}\\
 &\leq N^2e^{-n/16}\\
 &\leq \frac14e^{n/32}e^{-n/16}
   =\frac14e^{-n/32}<1.
\end{align*}
Therefore, with positive probability no pair is bad, and a family with the
required intersection property exists. On this event the sampled sets are
automatically distinct, since two equal balanced sets would have intersection
$n/2$. Finally, $n\geq512$ implies $e^{n/64}\geq e^8>6$, and hence
\[
 N=\left\lfloor\frac12e^{n/64}\right\rfloor
 \geq \frac13e^{n/64}.
\]
\end{proof}

For a balanced set $S$, write $G_S$ for the complete bipartite graph between $S$ and $\bar S$, with every edge assigned weight
\[
 w_0:=\frac{1}{64n\epsilon}.
\]
There are $|S||\bar S|=n^2/4$ such edges, and the cut $(S,\bar S)$
cuts all of them. Consequently,
\[
 \OPT(G_S)=\frac{w_0n^2}{4}=\frac{n}{256\epsilon}.
\]
Represent the support $S$ by $s\in\{\pm1\}^n$, where $s_i=+1$ exactly
when $i\in S$, and represent a cut $R$ similarly by
$r\in\{\pm1\}^n$. Since $S$ is balanced, $\sum_i s_i=0$.

\begin{lemma}\label{lem:cut_not_both_large}
For distinct $S,T\in\mathcal A$, the two sets of cuts
\[
 B_S:=\{R:\Phi_{G_S}(R)>\tfrac78\OPT(G_S)\},
 \qquad
 B_T:=\{R:\Phi_{G_T}(R)>\tfrac78\OPT(G_T)\}
\]
are disjoint.
\end{lemma}
\begin{proof}
Let $t\in\{\pm1\}^n$ be the sign vector of $T$. For an unordered pair
$\{i,j\}$, the indicator that it is an edge of $G_S$ is
$(1-s_is_j)/2$, while the indicator that the cut $R$ separates its endpoints
is $(1-r_ir_j)/2$. Therefore
\begin{align}
 \Phi_{G_S}(R)
 &=\frac{w_0}{4}\sum_{i<j}(1-s_is_j)(1-r_ir_j)\notag\\
 &=\frac{w_0}{4}\left[
    \binom n2-\sum_{i<j}s_is_j-\sum_{i<j}r_ir_j
    +\sum_{i<j}(s_ir_i)(s_jr_j)
   \right].
 \label{eq:weighted-cut-expansion}
\end{align}
For any sign vector $a\in\{\pm1\}^n$,
\[
 \sum_{i<j}a_ia_j
 =\frac12\left(\left(\sum_i a_i\right)^2-n\right).
\]
Apply this identity first to $s$, then to $r$, and finally to the sign vector
whose $i$th coordinate is $s_ir_i$. Since $\sum_i s_i=0$ and
$\sum_i s_ir_i=\langle s,r\rangle$, equation
\eqref{eq:weighted-cut-expansion} becomes
\[
 \Phi_{G_S}(R)
 =\frac{w_0}{8}\left(n^2+\langle s,r\rangle^2-\left(\sum_i r_i\right)^2\right).
\]
Suppose $R\in B_S$. Since
$\OPT(G_S)=w_0n^2/4$, the definition of $B_S$ and the last display give
\begin{align*}
 \frac{w_0}{8}\left(n^2+\langle s,r\rangle^2
              -\left(\sum_i r_i\right)^2\right)
 &>\frac78\cdot\frac{w_0n^2}{4},
\end{align*}
and hence
\[
 \langle s,r\rangle^2-\left(\sum_i r_i\right)^2
 >\frac{3n^2}{4}.
\]
In particular,
$|\langle s,r\rangle|>\sqrt3\,n/2$. For sign vectors, Hamming distance
satisfies
\[
 d_H(r,s)=\frac{n-\langle r,s\rangle}{2}.
\]
Choosing the closer of $s$ and $-s$, we obtain
\[
 \min_{\sigma\in\{\pm1\}}d_H(r,\sigma s)
 =\frac{n-|\langle r,s\rangle|}{2}
 <\frac{2-\sqrt3}{4}\,n.
\]
Set $\rho=(2-\sqrt3)/4$; note that $\rho<1/8$.

It remains to compare the antipodal pairs $\{s,-s\}$ and $\{t,-t\}$.
Writing $a=|S\cap T|$ and using $|S|=|T|=n/2$, we have
\[
 d_H(s,t)=|S\triangle T|=n-2a.
\]
Since $a\in(n/8,3n/8)$, it follows that
\[
 \frac n4<d_H(s,t)<\frac{3n}{4}.
\]
Moreover, $d_H(s,-t)=n-d_H(s,t)$, so it lies in the same interval. Thus
every vector in $\{s,-s\}$ is at Hamming distance greater than $n/4$ from
every vector in $\{t,-t\}$.

If a cut $R$ belonged to both $B_S$ and $B_T$, the preceding argument would
give signs $\sigma,\tau\in\{\pm1\}$ such that
\[
 d_H(r,\sigma s)<\rho n
 \qquad\text{and}\qquad
 d_H(r,\tau t)<\rho n.
\]
The triangle inequality would then imply
\[
 d_H(\sigma s,\tau t)<2\rho n<\frac n4,
\]
contradicting the separation just proved. Hence $B_S\cap B_T=\varnothing$.
\end{proof}

\begin{proof}[Proof of Theorem~\ref{thm:lower_bound}]
Set
\[
 W:=\OPT(G_S)=\frac{n}{256\epsilon};
\]
this value is the same for every $S\in\mathcal A$. Suppose, toward a
contradiction, that an $(\epsilon,0)$-DP algorithm $\mathcal M$ has expected
additive error strictly less than $n/(4096\epsilon)=W/16$ on every weighted
graph. For a fixed $S\in\mathcal A$, define the nonnegative random variable
\[
 L_S:=W-\Phi_{G_S}(\mathcal M(G_S)).
\]
By assumption, $\E{L_S}<W/16$. The event
$\mathcal M(G_S)\notin B_S$ is exactly the event $L_S\geq W/8$.
Markov's inequality therefore gives
\[
 \Pr{\mathcal M(G_S)\notin B_S}
 =\Pr{L_S\geq W/8}
 \leq\frac{\E{L_S}}{W/8}<\frac12.
\]
Equivalently,
\[
 \Pr{\mathcal M(G_S)\in B_S}>\frac12
 \qquad(S\in\mathcal A).
\]

We next compare the output distributions on two packing instances. Let
$w^S,w^T$ be the weight vectors of $G_S,G_T$. Each vector has $\ell_1$ norm
equal to its total edge weight $W$, and hence
\[
 D_{S,T}:=\|w^S-w^T\|_1
 \leq\|w^S\|_1+\|w^T\|_1
 =2W=\frac{n}{128\epsilon}.
\]
Let $K=\max\{1,\lceil D_{S,T}\rceil\}$. The line segment from $w^S$ to $w^T$ can be
divided into $K$ subsegments of $\ell_1$ length at most one: explicitly, use
the weight vectors
\[
 w^{(j)}=\left(1-\frac jK\right)w^S+\frac jK w^T,
 \qquad j=0,1,\ldots,K.
\]
These vectors remain nonnegative, and consecutive vectors are adjacent in the
weighted $\ell_1$ model. Applying pure differential privacy successively
along this chain gives, for every event $B$,
\[
 \Pr{\mathcal M(G_S)\in B}
 \geq e^{-K\epsilon}\Pr{\mathcal M(G_T)\in B}.
\]
Since $K\leq D_{S,T}+1$ and
$\epsilon D_{S,T}\leq n/128$, we obtain the convenient uniform bound
\[
 \Pr{\mathcal M(G_S)\in B}
 \geq e^{-n/128-\epsilon}\Pr{\mathcal M(G_T)\in B}.
\]

Now fix one reference set $S\in\mathcal A$. For every $T\in\mathcal A$,
apply the last display with $B=B_T$ and use the previously established
inequality on input $G_T$:
\[
 \Pr{\mathcal M(G_S)\in B_T}
 >\frac12e^{-n/128-\epsilon}.
\]
The events $B_T$, $T\in\mathcal A$, are pairwise disjoint by
Lemma~\ref{lem:cut_not_both_large}. Consequently, under the single output
distribution $\mathcal M(G_S)$,
\[
\begin{aligned}
 1
 &\geq \Pr{\mathcal M(G_S)\in\bigcup_{T\in\mathcal A}B_T}\\
 &=\sum_{T\in\mathcal A}\Pr{\mathcal M(G_S)\in B_T}\\
 &>\frac12e^{-n/128-\epsilon}|\mathcal A|\\
 &\geq\frac16e^{-n/128-\epsilon}e^{n/64}
 =\frac16e^{n/128-\epsilon}.
\end{aligned}
\]
Finally, $\epsilon\leq n/256$ implies
$n/128-\epsilon\geq n/256$, and $n\geq512$ gives
\[
 \frac16e^{n/128-\epsilon}
 \geq\frac16e^{n/256}
 \geq\frac{e^2}{6}>1.
\]
This leads to a contradiction. Thus no such
private algorithm can have expected error strictly below
$n/(4096\epsilon)$ on every input.
\end{proof}

\section{Conclusion and Open Questions}

In this work, we instantiate the general study of approximation algorithms for Max-CSPs under differential privacy where the privacy protection is at the constraint level. While we have made some initial progresses towards understanding this question, our study also leaves several important open questions, which we highlight below.
\begin{itemize}
\item \textbf{Closing Gaps for All CSPs: } In the (more interesting) high-privacy regime ($\eps \leq 1$), we have established an upper bound (impossibility result) of $\mu + O(\eps)$ on the approximation ratio. However, we only present algorithms that (nearly) match this bound only in certain cases (such as Max-Cut and Max-$k$XOR for odd $k$). The obvious open question here is to devise an algorithm that works for \emph{all} Max-CSPs and matches the $\mu + O(\eps)$ upper bound.
\item \textbf{Computational Efficiency: } As discussed in the introduction, it is in general impossible to hope for a polynomial-time algorithm for Max-CSPs with an approximation ratio better than random (i.e. $\mu$) since this is NP-hard for certain CSPs (e.g.~\cite{Chan13}). Nevertheless, some other CSPs are known to admit better-than-$\mu$ approximation in polynomial time, including the celebrated Goemans-Williamson algorithm for Max-Cut~\cite{GoemansW94}. For such CSPs, it would be interesting to achieve an efficient better-than-$\mu$ approximation with DP as well.
\item \textbf{Low-Privacy Regime: } Although the low-privacy regime $(\eps > 1)$ is generally less studied in the literature, we nevertheless point out that there is still an exponential gap on the ``deficit'' in our upper bound of $1 - O(1/e^{\eps})$ and lower bound of $1 - O(1/\eps)$. It remains an interesting question to close this gap.
\end{itemize}
\section*{Acknowledgments}

We are grateful to Jittat Fakcharoenphol for insightful discussions in the early stages of this work.

\bibliographystyle{alpha}
\bibliography{privacy}

\newcommand{\etalchar}[1]{$^{#1}$}
\begin{thebibliography}{CAEL{\etalchar{+}}22}

\bibitem[ACD{\etalchar{+}}25]{aamand2025breaking}
Anders Aamand, Justin~Y. Chen, Mina Dalirrooyfard, Slobodan Mitrovi{\'c}, Yuriy
  Nevmyvaka, Sandeep Silwal, and Yinzhan Xu.
\newblock Breaking the \$n{\textasciicircum}\{1.5\}\$ additive error barrier
  for private and efficient graph sparsification via private expander
  decomposition.
\newblock In {\em Forty-second International Conference on Machine Learning},
  2025.

\bibitem[AU19]{arora2019differentially}
Raman Arora and Jalaj Upadhyay.
\newblock On differentially private graph sparsification and applications.
\newblock In {\em Advances in Neural Information Processing Systems}, pages
  13378--13389, 2019.

\bibitem[BBDS12]{blocki2012johnson}
Jeremiah Blocki, Avrim Blum, Anupam Datta, and Or~Sheffet.
\newblock The {J}ohnson-{L}indenstrauss transform itself preserves differential
  privacy.
\newblock In {\em Foundations of Computer Science (FOCS), 2012 IEEE 53rd Annual
  Symposium on}, pages 410--419. IEEE, 2012.

\bibitem[BBG18]{balle2018privacy}
Borja Balle, Gilles Barthe, and Marco Gaboardi.
\newblock Privacy amplification by subsampling: Tight analyses via couplings
  and divergences.
\newblock In Samy Bengio, Hanna~M. Wallach, Hugo Larochelle, Kristen Grauman,
  Nicol{\`{o}} Cesa{-}Bianchi, and Roman Garnett, editors, {\em Advances in
  Neural Information Processing Systems 31: Annual Conference on Neural
  Information Processing Systems 2018, NeurIPS 2018, 3-8 December 2018,
  Montr{\'{e}}al, Canada}, pages 6280--6290, 2018.

\bibitem[BEK21]{bun2021differentially}
Mark Bun, Marek Elias, and Janardhan Kulkarni.
\newblock Differentially private correlation clustering.
\newblock In {\em International Conference on Machine Learning}, pages
  1136--1146. PMLR, 2021.

\bibitem[BMO{\etalchar{+}}15]{BarakMORRSTVWW15}
Boaz Barak, Ankur Moitra, Ryan O'Donnell, Prasad Raghavendra, Oded Regev, David
  Steurer, Luca Trevisan, Aravindan Vijayaraghavan, David Witmer, and John
  Wright.
\newblock Beating the random assignment on constraint satisfaction problems of
  bounded degree.
\newblock In Naveen Garg, Klaus Jansen, Anup Rao, and Jos{\'{e}} D.~P. Rolim,
  editors, {\em Approximation, Randomization, and Combinatorial Optimization.
  Algorithms and Techniques, {APPROX/RANDOM} 2015, August 24-26, 2015,
  Princeton, NJ, {USA}}, volume~40 of {\em LIPIcs}, pages 110--123. Schloss
  Dagstuhl - Leibniz-Zentrum f{\"{u}}r Informatik, 2015.

\bibitem[BPS99]{brailsford1999constraint}
Sally~C Brailsford, Chris~N Potts, and Barbara~M Smith.
\newblock Constraint satisfaction problems: Algorithms and applications.
\newblock {\em European journal of operational research}, 119(3):557--581,
  1999.

\bibitem[CAEL{\etalchar{+}}22]{cohen2022scalable}
Vincent Cohen-Addad, Alessandro Epasto, Silvio Lattanzi, Vahab Mirrokni, Andres
  Munoz~Medina, David Saulpic, Chris Schwiegelshohn, and Sergei Vassilvitskii.
\newblock Scalable differentially private clustering via hierarchically
  separated trees.
\newblock In {\em Proceedings of the 28th ACM SIGKDD Conference on Knowledge
  Discovery and Data Mining}, pages 221--230, 2022.

\bibitem[CDFZ24]{chandra2024differentially}
Rishi Chandra, Michael Dinitz, Chenglin Fan, and Zongrui Zou.
\newblock Differentially private algorithms for graph cuts: A shifting
  mechanism approach and more.
\newblock {\em arXiv preprint arXiv:2407.06911}, 2024.

\bibitem[CFL{\etalchar{+}}22]{DBLP:conf/nips/Cohen-AddadFLMN22}
Vincent Cohen{-}Addad, Chenglin Fan, Silvio Lattanzi, Slobodan Mitrovic, Ashkan
  Norouzi{-}Fard, Nikos Parotsidis, and Jakub Tarnawski.
\newblock Near-optimal correlation clustering with privacy.
\newblock In Sanmi Koyejo, S.~Mohamed, A.~Agarwal, Danielle Belgrave, K.~Cho,
  and A.~Oh, editors, {\em Advances in Neural Information Processing Systems
  35: Annual Conference on Neural Information Processing Systems 2022, NeurIPS
  2022, New Orleans, LA, USA, November 28 - December 9, 2022}, 2022.

\bibitem[Cha13]{Chan13}
Siu~On Chan.
\newblock Approximation resistance from pairwise independent subgroups.
\newblock In Dan Boneh, Tim Roughgarden, and Joan Feigenbaum, editors, {\em
  Symposium on Theory of Computing Conference, STOC'13, Palo Alto, CA, USA,
  June 1-4, 2013}, pages 447--456. {ACM}, 2013.

\bibitem[CS97]{cheng1997applying}
Cheng-Chung Cheng and Stephen~F Smith.
\newblock Applying constraint satisfaction techniques to job shop scheduling.
\newblock {\em Annals of Operations Research}, 70(0):327--357, 1997.

\bibitem[DFKO07]{dinur2007fourier}
Irit Dinur, Ehud Friedgut, Guy Kindler, and Ryan O'Donnell.
\newblock On the fourier tails of bounded functions over the discrete cube.
\newblock {\em Israel Journal of Mathematics}, 160:389--412, 2007.

\bibitem[DGU{\etalchar{+}}25]{Deng0U0Z25}
Chengyuan Deng, Jie Gao, Jalaj Upadhyay, Chen Wang, and Samson Zhou.
\newblock On the price of differential privacy for hierarchical clustering.
\newblock In {\em The Thirteenth International Conference on Learning
  Representations, {ICLR}}, 2025.

\bibitem[DMN23]{DBLP:conf/nips/DalirrooyfardMN23}
Mina Dalirrooyfard, Slobodan Mitrovic, and Yuriy Nevmyvaka.
\newblock Nearly tight bounds for differentially private multiway cut.
\newblock In {\em Proceedings of the 37th International Conference on Neural
  Information Processing Systems}, pages 24947--24965, 2023.

\bibitem[DMNS06]{dwork2006calibrating}
Cynthia Dwork, Frank McSherry, Kobbi Nissim, and Adam Smith.
\newblock {Calibrating Noise to Sensitivity in Private Data Analysis}.
\newblock In {\em TCC}, pages 265--284, 2006.

\bibitem[DP09]{dubhashi2009concentration}
Devdatt~P Dubhashi and Alessandro Panconesi.
\newblock {\em Concentration of measure for the analysis of randomized
  algorithms}.
\newblock Cambridge University Press, 2009.

\bibitem[DR14]{dwork2014algorithmic}
Cynthia Dwork and Aaron Roth.
\newblock The algorithmic foundations of differential privacy.
\newblock {\em Foundations and Trends in Theoretical Computer Science},
  9(3--4):211--407, 2014.

\bibitem[DSZ{\etalchar{+}}21]{ding2021differentially}
Xiaofeng Ding, Shujun Sheng, Huajian Zhou, Xiaodong Zhang, Zhifeng Bao, Pan
  Zhou, and Hai Jin.
\newblock Differentially private triangle counting in large graphs.
\newblock {\em IEEE Transactions on Knowledge and Data Engineering},
  34(11):5278--5292, 2021.

\bibitem[EKKL20]{eliavs2020differentially}
Marek Eli{\'a}{\v{s}}, Michael Kapralov, Janardhan Kulkarni, and Yin~Tat Lee.
\newblock Differentially private release of synthetic graphs.
\newblock In {\em Proceedings of the Fourteenth Annual ACM-SIAM Symposium on
  Discrete Algorithms}, pages 560--578. SIAM, 2020.

\bibitem[FLP{\etalchar{+}}26]{fan2026private}
Chenglin Fan, Jingcheng Liu, Peng Pan, Hangyu Xu, and Zongrui Zou.
\newblock Private approximation of graph spectra and cuts via spectral
  amplifiers.
\newblock {\em arXiv preprint arXiv:2607.18846}, 2026.

\bibitem[Ghe13]{ghedira2013constraint}
Khaled Ghedira.
\newblock {\em Constraint satisfaction problems: csp formalisms and
  techniques}.
\newblock John Wiley \& Sons, 2013.

\bibitem[GRS09]{ghosh2009universally}
Arpita Ghosh, Tim Roughgarden, and Mukund Sundararajan.
\newblock Universally utility-maximizing privacy mechanisms.
\newblock In {\em Proceedings of the forty-first annual ACM symposium on Theory
  of computing}, pages 351--360, 2009.

\bibitem[GRU12]{gupta2012iterative}
Anupam Gupta, Aaron Roth, and Jonathan Ullman.
\newblock Iterative constructions and private data release.
\newblock In {\em Theory of cryptography conference}, pages 339--356. Springer,
  2012.

\bibitem[GW94]{GoemansW94}
Michel~X. Goemans and David~P. Williamson.
\newblock .879-approximation algorithms for {MAX} {CUT} and {MAX} 2sat.
\newblock In Frank~Thomson Leighton and Michael~T. Goodrich, editors, {\em
  Proceedings of the Twenty-Sixth Annual {ACM} Symposium on Theory of
  Computing, 23-25 May 1994, Montr{\'{e}}al, Qu{\'{e}}bec, Canada}, pages
  422--431. {ACM}, 1994.

\bibitem[H{\aa}s01]{Hastad01}
Johan H{\aa}stad.
\newblock Some optimal inapproximability results.
\newblock {\em J. {ACM}}, 48(4):798--859, 2001.

\bibitem[HR12]{hardt2012beating}
Moritz Hardt and Aaron Roth.
\newblock Beating randomized response on incoherent matrices.
\newblock In {\em Proceedings of the 44th Annual ACM Symposium on Theory of
  Computing}, pages 1255--1268, 2012.

\bibitem[HV02]{haastad2002advantage}
Johan H{\aa}stad and Srinivasan Venkatesh.
\newblock On the advantage over a random assignment.
\newblock In {\em Proceedings of the thiry-fourth annual ACM symposium on
  Theory of computing}, pages 43--52, 2002.

\bibitem[IEM{\etalchar{+}}23]{DBLP:conf/icml/ImolaEMCM23}
Jacob Imola, Alessandro Epasto, Mohammad Mahdian, Vincent Cohen{-}Addad, and
  Vahab Mirrokni.
\newblock Differentially private hierarchical clustering with provable
  approximation guarantees.
\newblock In Andreas Krause, Emma Brunskill, Kyunghyun Cho, Barbara Engelhardt,
  Sivan Sabato, and Jonathan Scarlett, editors, {\em International Conference
  on Machine Learning, {ICML} 2023, 23-29 July 2023, Honolulu, Hawaii, {USA}},
  volume 202 of {\em Proceedings of Machine Learning Research}, pages
  14353--14375. {PMLR}, 2023.

\bibitem[IMC22]{imola2022differentially}
Jacob Imola, Takao Murakami, and Kamalika Chaudhuri.
\newblock Differentially private triangle and 4-cycle counting in the shuffle
  model.
\newblock In {\em Proceedings of the 2022 ACM SIGSAC Conference on Computer and
  Communications Security}, pages 1505--1519, 2022.

\bibitem[KKMO07]{KhotKMO07}
Subhash Khot, Guy Kindler, Elchanan Mossel, and Ryan O'Donnell.
\newblock Optimal inapproximability results for {MAX-CUT} and other 2-variable
  csps?
\newblock {\em {SIAM} J. Comput.}, 37(1):319--357, 2007.

\bibitem[KN07]{khot2007linear}
Subhash Khot and Assaf Naor.
\newblock Linear equations modulo 2 and the l1 diameter of convex bodies.
\newblock In {\em 48th Annual IEEE Symposium on Foundations of Computer Science
  (FOCS'07)}, pages 318--328. IEEE, 2007.

\bibitem[KNRS13]{DBLP:conf/tcc/KasiviswanathanNRS13}
Shiva~Prasad Kasiviswanathan, Kobbi Nissim, Sofya Raskhodnikova, and Adam~D.
  Smith.
\newblock Analyzing graphs with node differential privacy.
\newblock In Amit Sahai, editor, {\em Theory of Cryptography - 10th Theory of
  Cryptography Conference, {TCC} 2013, Tokyo, Japan, March 3-6, 2013.
  Proceedings}, volume 7785 of {\em Lecture Notes in Computer Science}, pages
  457--476. Springer, 2013.

\bibitem[KZ97]{KarloffZ97}
Howard~J. Karloff and Uri Zwick.
\newblock A 7/8-approximation algorithm for {MAX} 3sat?
\newblock In {\em 38th Annual Symposium on Foundations of Computer Science,
  {FOCS} 1997, Miami Beach, Florida, USA, October 19-22, 1997}, pages 406--415.
  {IEEE} Computer Society, 1997.

\bibitem[LUZ24]{liu2024optimal}
Jingcheng Liu, Jalaj Upadhyay, and Zongrui Zou.
\newblock Optimal bounds on private graph approximation.
\newblock In {\em Proceedings of the 2024 Annual ACM-SIAM Symposium on Discrete
  Algorithms (SODA)}, pages 1019--1049. SIAM, 2024.

\bibitem[LUZ25]{liu2025almost}
Jingcheng Liu, Jalaj Upadhyay, and Zongrui Zou.
\newblock Almost linear time differentially private release of synthetic
  graphs.
\newblock In {\em International Conference on Artificial Intelligence and
  Statistics}, pages 289--297. PMLR, 2025.

\bibitem[MT07]{mcsherry2007mechanism}
Frank McSherry and Kunal Talwar.
\newblock Mechanism design via differential privacy.
\newblock In {\em 48th Annual IEEE Symposium on Foundations of Computer Science
  (FOCS'07)}, pages 94--103. IEEE, 2007.

\bibitem[O'D14]{odonnell2014analysis}
Ryan O'Donnell.
\newblock {\em Analysis of Boolean Functions}.
\newblock Cambridge University Press, 2014.

\bibitem[Rag08]{Raghavendra08}
Prasad Raghavendra.
\newblock Optimal algorithms and inapproximability results for every csp?
\newblock In Cynthia Dwork, editor, {\em Proceedings of the 40th Annual {ACM}
  Symposium on Theory of Computing, Victoria, British Columbia, Canada, May
  17-20, 2008}, pages 245--254. {ACM}, 2008.

\bibitem[RS09]{RaghavendraS09a}
Prasad Raghavendra and David Steurer.
\newblock How to round any {CSP}.
\newblock In {\em 50th Annual {IEEE} Symposium on Foundations of Computer
  Science, {FOCS} 2009, Atlanta, Georgia, USA, October 25-27, 2009}, pages
  586--594. {IEEE} Computer Society, 2009.

\bibitem[SGB08]{salido2008introduction}
Miguel~A Salido, Antonio Garrido, and Roman Bart{\'a}k.
\newblock Introduction: Special issue on constraint satisfaction techniques for
  planning and scheduling problems.
\newblock {\em Engineering Applications of Artificial Intelligence},
  21(5):679--682, 2008.

\bibitem[She92]{shearer1992note}
James~B Shearer.
\newblock A note on bipartite subgraphs of triangle-free graphs.
\newblock {\em Random Structures \& Algorithms}, 3(2):223--226, 1992.

\bibitem[SPHH25]{suppakitpaisarn2025counting}
Vorapong Suppakitpaisarn, Donlapark Ponnoprat, Nicha Hirankarn, and Quentin
  Hillebrand.
\newblock Counting graphlets of size k under local differential privacy.
\newblock In {\em International Conference on Artificial Intelligence and
  Statistics}, pages 5005--5013. PMLR, 2025.

\bibitem[ST09]{SamorodnitskyT09}
Alex Samorodnitsky and Luca Trevisan.
\newblock Gowers uniformity, influence of variables, and pcps.
\newblock {\em {SIAM} J. Comput.}, 39(1):323--360, 2009.

\bibitem[US19]{DBLP:conf/nips/UllmanS19}
Jonathan~R. Ullman and Adam Sealfon.
\newblock Efficiently estimating erdos-renyi graphs with node differential
  privacy.
\newblock In Hanna~M. Wallach, Hugo Larochelle, Alina Beygelzimer, Florence
  d'Alch{\'{e}}{-}Buc, Emily~B. Fox, and Roman Garnett, editors, {\em Advances
  in Neural Information Processing Systems 32: Annual Conference on Neural
  Information Processing Systems 2019, NeurIPS 2019, December 8-14, 2019,
  Vancouver, BC, Canada}, pages 3765--3775, 2019.

\bibitem[UUA21]{upadhyay2021differentially}
Jalaj Upadhyay, Sarvagya Upadhyay, and Raman Arora.
\newblock Differentially private analysis on graph streams.
\newblock In {\em International Conference on Artificial Intelligence and
  Statistics}, pages 1171--1179. PMLR, 2021.

\end{thebibliography}

\clearpage
\appendix

\end{document}